\documentclass[a4paper,twocolumn,10pt,accepted=2024-11-26]{quantumarticle}
\pdfoutput=1
\usepackage[utf8]{inputenc}
\usepackage[english]{babel}
\usepackage[T1]{fontenc}
\usepackage{amsmath}
\usepackage{amsthm}
\usepackage{amssymb}
\usepackage{amsfonts}
\usepackage{comment}
\usepackage{mathbbol}
\usepackage{hyperref}
\hypersetup{
    colorlinks=true,
    urlcolor=blue,
	citecolor=blue,
    allcolors=blue
}
\usepackage{bm}
\PassOptionsToPackage{compress}{natbib}
\usepackage[numbers,sort&compress]{natbib}
\usepackage{setspace}

\newcommand{\ket}[1]{{| {#1} \rangle}}
\newcommand{\bra}[1]{{\left\langle {#1} \right|}}
\newcommand{\ketbra}[2]{{\left| {#1} \right\rangle \!\!\left\langle {#2} \right|}}

\newcommand{\ignore}[1]{}

\DeclareMathOperator{\tr}{tr}
\DeclareMathOperator{\sgn}{sgn}

\newtheorem{proposition}{Proposition}

\newtheorem{definition}{Definition}
\newtheorem{conjecture}{Conjecture}

\begin{document}

\title{Polytopes of Absolutely Wigner Bounded Spin States}
\author{J\'{e}r\^{o}me Denis}
\affiliation{Institut de Physique Nucl\'{e}aire, Atomique et de Spectroscopie, CESAM, University of Li\`{e}ge, B-4000 Li\`{e}ge, Belgium}
\author{Jack Davis}
\affiliation{Department of Physics and Astronomy, University of Waterloo, Waterloo, Ontario, Canada N2L 3G1}
\affiliation{Institute for Quantum Computing, University of Waterloo, Waterloo, Ontario, Canada N2L 3G1}
\affiliation{DIENS, \'Ecole Normale Sup\'erieure, PSL University, CNRS, INRIA, 45 rue d'Ulm, Paris 75005, France}
\author{Robert B. Mann}
\affiliation{Department of Physics and Astronomy, University of Waterloo, Waterloo, Ontario, Canada N2L 3G1}
\affiliation{Institute for Quantum Computing, University of Waterloo, Waterloo, Ontario, Canada N2L 3G1}
\affiliation{Perimeter Institute for Theoretical Physics, 31 Caroline St N, Waterloo, Ontario, Canada N2L 2Y5}
\author{John Martin}
\affiliation{Institut de Physique Nucl\'{e}aire, Atomique et de Spectroscopie, CESAM, University of Li\`{e}ge, B-4000 Li\`{e}ge, Belgium}

\maketitle

\begin{abstract}
Quasiprobability has become an increasingly popular notion for characterising non-classicality in quantum information, thermodynamics, and metrology.  Two important distributions with non-positive quasiprobability are the Wigner function and the Glauber-Sudarshan function.  Here we study properties of the spin Wigner function for finite-dimensional quantum systems and draw comparisons with its infinite-dimensional analog, focusing in particular on the relation to the Glauber-Sudarshan function and the existence of absolutely Wigner-bounded states.  More precisely, we investigate unitary orbits of mixed spin states that are characterized by Wigner functions lower-bounded by a specified value. To this end, we extend a characterization of the set of absolutely Wigner positive states as a set of linear eigenvalue constraints, which together define a polytope centred on the maximally mixed state in the simplex of spin-$j$ states. The lower bound determines the relative size of such absolutely Wigner bounded (AWB) polytopes and we study their geometric characteristics. In each dimension a Hilbert-Schmidt ball representing a tight purity-based sufficient condition to be AWB is exactly determined, while another ball representing a necessary condition to be AWB is conjectured. Special attention is given to the case where the polytope separates orbits containing only positive Wigner functions from other orbits because of the use of Wigner negativity as a witness of non-classicality. Comparisons are made to absolute symmetric state separability and spin Glauber-Sudarshan positivity, with additional details given for low spin quantum numbers.
\end{abstract}

\section{Introduction}

When studying many-body mixed states,  the question naturally arises as to the maximal amount of entanglement that can be generated under arbitrary unitary evolution.  As state purity decreases due to interaction with the environment, such maximal entanglement is expected to decrease as well, with there being some point beyond which unitary evolution alone will not suffice to create entanglement. This has led to the concept of absolute separability, a defining property of certain mixed states whose isolated evolution cannot yield any entanglement.  In the case of $N$ qubits in a symmetric state, or equivalently a single system with spin $j=N/2$, recent work, including one based on the results presented here, has been devoted to characterizing the non-trivial set of symmetric absolutely separable (SAS) states~\cite{Champagne2022, SerranoEnsstiga2023, 2024Serrano}. In this context, SAS state balls centred on the maximally mixed state have been found. Inspired by this research on entanglement, we study here another measure of non-classicality given by the presence of negative values of the spin Wigner function for a finite-dimensional system.

Wigner negativity has long been an indicator of non-classicality in quantum systems. It is a necessary feature to observe a Bell-type violation with phase-space observables, see, e.g., \cite{Acn2009} and references therein. For pure states, Hudson’s theorem identifies Wigner negativity as both a necessary and sufficient condition for non-Gaussianity. In the setting of quantum information science, this connection is further reinforced by its pivotal role in enabling quantum advantage, particularly within the magic state injection model of universal fault-tolerant quantum computation~\cite{Veitch_Ferrie_Gross_Emerson_2012, mari_eisert_simulation_2012, Howard_Wallman_Veitch_Emerson_2014, PashayanEstimating2015, Delfosse_qudits_2017, Schmid_Unique_2022, booth_cv_connection_2022}. The associated resource theories \cite{Veitch_resource_2014, Albarelli_cv_resource_2018, Takagi_resource_2018, Wang_magic_channels_2019} further demonstrate that more negativity volume may be associated with more non-classicality. Parallel to these studies, in the field of many-body quantum dynamics of large interacting systems, the spin Wigner function \cite{stratonovich_distributions_1956} was recently used to accurately model the dynamics of spin systems~\cite{2022Fleischhauer,2023Fleischhauer}. Under certain dynamics, dephasing for example, the evolution of a spin state with an initially everywhere positive Wigner function can be simulated efficiently using stochastic trajectories.  This connection to the spin Wigner function, together with the growing prominence of Wigner negativity as a whole, shows that a classification of general mixed states in relation to their Wigner negativity is needed.

In this article, we completely characterize the degree of Wigner negativity that can be obtained by a set of states sharing a given spectrum.  Analogous to SAS states, we call a spin-$j$ state \textit{Absolutely Wigner-Positive} (AWP) if its spin Wigner function remains positive everywhere under the action of all unitaries $U\in\text{SU}(2j+1)$. From the point of view of a spin-$j$ system composed of $N=2j$ spin-$\frac{1}{2}$ objects (i.e., qubits), these unitary transformations are those that correspond to the most general unitary evolution of the $N$ spins in the symmetric subspace, i.e., those that connect any two states with the same spectrum in the symmetric subspace.

A natural consequence of this effort is a deeper understanding of whether central results from the original phase-space picture for continuous variable systems on the connection between Glauber-Sudarshan positivity and Wigner positivity still hold in the SU(2)-based phase-space picture of finite-dimensional systems.  Despite the strong formal relationship between the planar (infinite-dimensional) and spherical (finite-dimensional) descriptions \cite{brif_mann_lie_1999} relatively less is known about the latter.  Investigating such differences hence forms a second objective of this work.

In order to position our work in a wider context, we begin with a brief note on related research. Recent investigations have studied AWP states for a wide class of Wigner functions in general $d$-dimensional quantum systems \cite{Abbasli2020, Abgaryan2021, Abgaryan2020, Abgaryan2021b}.  In particular, the set of AWP states with respect to a given Wigner function was found to form a polytope within the simplex of state spectra.  It was further shown that this class of functions always contains the particular one introduced by Stratonovich \cite{stratonovich_distributions_1956}, rediscovered by Agarwal \cite{1981Agarwal, dowling_agarwal_1994}, and now used in many experimental set-ups \cite{Riedel_Nature_2010, Schmied_Treutlein_NJP_2011, McConnell_3000atoms_Nature_2015, Chen_spin_in_diamond_2019}; see \cite{Klimov_Romero_de_Guise_2017} for a recent review.  By contrast, here we focus exclusively on this canonical Wigner function for spin-$j$ systems, which is the only SU(2)-covariant Wigner function compatible with other well-known quasiprobability distributions in the following senses: (i) it can be continuously transformed into the Husimi function or into the Glauber-Sudarshan function~\cite{varilly_moyal_1989}, and (ii) the original Wigner function defined for infinite dimensional quantum systems is retrieved by taking the infinite spin limit $j\rightarrow\infty$ \cite{Weigert_contracting_2000}.  In addition to offering a related but alternative derivation of the polytopes first discovered in \cite{Abbasli2020, Abgaryan2021, Abgaryan2020, Abgaryan2021b}, here we go beyond previous investigations in three ways. The first is that we extend the argument to include
unitary orbits of Wigner functions lower-bounded by values that are not necessarily zero.
These one-parameter families of polytopes, which we refer to as \textit{absolutely Wigner bounded} (AWB) polytopes, are of interest not only for Wigner functions but also for other quasiprobability distributions. The second is that we go into explicit detail on the geometric properties of these polytopes in all dimensions, including an analysis on their infinite-spin limit\ignore{, and explore their relevance in the context of simulating quantum spin systems with the help of classical methods}. The third is that we compare the AWP polytopes with the set of SAS states, which amounts to a comparison between Wigner negativity and entanglement in the mixed state setting.

Our first result is the complete characterization of the set of AWB spin states in all finite dimensions, with AWP states appearing as a very special case. In particular, the set of AWB states forms a polytope in the simplex of density matrix spectra. The polytopes are delimited
by $(2j+1)!$ hyperplanes defined by permutations of the eigenvalues of the kernel operator defining the Wigner function. Centred on the maximally mixed state for each dimension, we also exactly find the largest possible ball containing nothing but AWB states, which amounts
to the tightest sufficient condition to be AWB based solely on the purity of mixed states. We also obtain an expression that we conjecture to describe the smallest ball containing all AWB states, which amounts to the tightest necessary condition based solely on the purity of mixed states. Numerical evidence strongly supports this conjecture. For both criteria, we discuss their geometric
interpretation in relation to the full AWB polytope. We then specialize to the set of AWP states and compare them with the set of SAS states \cite{Giraud_classicality_2008,Bohnet-Waldraff-PPT_2016,Bohnet-Waldraff2017absolutely}.
We emphasize that, as there is a continuous transformation relating the continuous Wigner function and the other phase space functions, the results obtained here are completely transferable to the Husimi and Glauber-Sudarshan functions~\cite{2024Serrano}.

We then use these polytopes to draw comparisons between the phase space description of finite-dimensional quantum systems with that of infinite-dimensional ones.  In particular, we proved the existence of states that are each Glauber-Sudarshan-positive (even absolutely positive) yet Wigner-negative.  This is in stark contrast to the bosonic setting where a positive Glauber-Sudarshan function trivially implies a positive Wigner function due to their well-known relationship though Gaussian convolution \cite{Cahill_Glauber_quasi_1969, Lee_nonclassical_depth_1991}.  Conversely, we show that there exists a non-trivial set of states that are each Wigner-positive (even AWP) yet are Glauber-Sudarshan-negative.

Finally we analyze the infinite spin limit of these polytopes in the context of the well-known spin-to-boson contraction \cite{Arecchi_1972} and its manifestation on the level of Wigner functions \cite{Weigert_contracting_2000, koczor_parity_2020}.  In particular, the volumes of the AWP balls vanish in the limit $j\to\infty$, which offers strong evidence that the notion of AWP states cannot exist in the infinite-dimensional setting. Furthermore, the outer AWB ball for non-zero cutoff is not found to vanish, implying that bosonic AWB states may indeed exist.

Our paper is organized as follows. Section \ref{sec:background} briefly outlines the generalized phase space picture using the parity-operator/Stratonovich framework for the group SU(2).  Section \ref{sec:AWPPolytope} proves general results on AWB polytopes, valid for any spin quantum number. More precisely, Secs.\ \ref{Sec:AWBpolytopes}--\ref{Sec:Majorization} derive and characterize AWB polytopes, Sec.\ \ref{sec:AWP_balls} determines and conjectures, respectively, the largest and smallest Hilbert-Schmidt ball sitting inside and outside the AWB polytopes, and Sec.\ \ref{sec:infinite-spin-limit} studies the infinite spin limit. Section \ref{sec:entanglement} explores low-dimensional cases in more detail and draws comparisons to entanglement.  Finally, conclusions are drawn and perspectives are outlined in Sec.~\ref{sec:conclusion}. The manuscript ends with appendices containing technical developments.

\section{Background}
\label{sec:background}

The parity-operator framework is the generalization of Moyal quantum mechanics to physical systems other than a collection of non-relativistic spinless particles; see \cite{2021Everitt} for a recent information-theoretic review.  Each type of system has a different phase space, and the various types are classified by the system's dynamical symmetry group \cite{brif_mann_lie_1999}.  In each case the central object is a map, $\Delta$, called the \textit{kernel}, which takes points in phase space to operators on Hilbert space.  A quasi-probability representation of a quantum state, evaluated at a point in phase space, is the expectation value of the phase-point operator assigned to that point.  Different kernels yield different distributions but all must obey the Stratonovich-Weyl axioms, which ensure, among other properties, the existence of an inverse map and that the Moyal picture is as close as possible to classical statistical physics over the same phase space (i.e., the Born rule as an $L^2$ inner product).

When applied to the Heisenberg-Weyl group (i.e., the group of displacement operators generated by the canonical commutation relations, $[ x,p ] = i \mathbb{1}$) this framework reduces to the common phase space associated with $n$ canonical degrees of freedom, $\mathbb{R}^{2n}$, and the phase-point operators corresponding to the Wigner function appear as a set of displaced parity operators \cite{brif_mann_lie_1999, Grossmann_1976, royer_parity_1977}.  A spin-$j$ system on the other hand corresponds to the group SU(2), which yields a spherical phase space, $S^2$.  Here we list some necessary results from this case; see the recent review \cite{Klimov_Romero_de_Guise_2017} and references therein for more information.

\subsection{Wigner function of a spin state}
Consider a single spin system with spin quantum number $j$.  Pure states live in the Hilbert space $\mathcal{H}\simeq \mathbb{C}^{2j+1}$, which carries an irreducible unitary representation of SU(2), $U_g$ for $g\in \text{SU}(2)$.
Mixed states live in the space of operators, $\mathcal{L(H)}$, where SU(2) acts via conjugation: $U_g \rho U^\dagger_g$.  This action on operator space is not irreducible and may be conveniently decomposed into irreducible multipoles.

The SU(2) Wigner kernel of a spin-$j$ system is
\begin{equation}\label{Wignerkernel}
\begin{aligned}
    &\Delta: S^2\rightarrow \mathcal{L(H)} \\
    &\Delta(\Omega)= \sqrt{\frac{4\pi}{2j+1}}\sum_{L=0}^{2j}\sum_{M=-L}^{L}Y_{LM}^{*}(\Omega)T_{LM},
\end{aligned}
\end{equation}
where $\Omega = (\theta,\phi) \in S^2$, $Y_{LM}(\Omega)$ are the spherical harmonics, and $T_{LM} \equiv T_{LM}^{(j)}$ are the spherical tensor operators associated with spin $j$ \cite{Varshalovich_1988}.  To avoid cluttered notation we do not label the operator $\Delta$ with a $j$; the surrounding context should be clear on which dimension/spin is being discussed.  The Wigner function of a spin state $\rho$ is defined as
\begin{equation}
    \begin{aligned}
    W_\rho(\Omega) & = \mathrm{Tr}\left[\rho\Delta(\Omega)\right]\\
     &= \frac{1}{2j+1} + \sqrt{\frac{4\pi}{2j+1}}\sum_{L=1}^{2j}\sum_{M=-L}^{L}\rho_{LM}Y_{LM}(\Omega),
     \label{eq:defWignerFunction}
    \end{aligned}
\end{equation}
where $\rho_{LM} = \tr[\rho \, T^\dagger_{LM}]$ are state multipoles~\cite{1981Agarwal}.  This function is normalized according to
\begin{equation}\label{eq:normalization}
    \frac{2j+1}{4\pi} \int_{S^2} W_\rho(\Omega) \, d\Omega = 1,
\end{equation}
and, as Eq.\ \eqref{eq:defWignerFunction} suggests, the maximally mixed state (MMS) $\rho_0 = \mathbb{1}/(2j+1)$ is mapped to the constant function 
\begin{equation}
    W_{\rho_0}(\Omega) = \frac{1}{2j+1}.
\end{equation}
Note that, using the Condon-Shortley phase convention, the Wigner function is real-valued for Hermitian operators, and in particular for quantum states.

An important property is SU(2) covariance:
\begin{equation}\label{eq:covariance}
    W_{U_g \rho U^\dagger_g}(\Omega) = W_{\rho}(g^{-1}\, \Omega),
\end{equation}
where the right hand side denotes the spatial action of SU(2) on the sphere.  As this is simply a rigid rotation, analogous to an optical displacement operator rigidly translating $\mathbb{R}^{2n}$, the overall \ignore{functional form} shape of any Wigner function is unaffected (i.e., the graph of the function is fixed up to rotation).  Hence the Wigner negative volume is defined as~\cite{2021Everitt,davis2021wigner}
\begin{equation}
\label{eq:WignerNeg}
    \delta(\rho)=\frac{1}{2}\left(\int_{\Gamma}\left|W_{\rho}(\Omega)\right|d\mu(\Omega)-1\right),
\end{equation}
often used as a quantification of Wigner negativity and a measure of non-classicality, is invariant under SU(2) transformations. Note that the action of a general unitary $U \in $ SU$(2j+1)$ on a state $\rho$ can of course radically change its Wigner function and thus also its negative volume. The quantity $d\mu(\Omega) = (2j+1)/(4\pi) \sin\theta d\theta d\phi$ is the invariant measure on the phase space.

A related consequence of SU(2) covariance is that all phase-point operators have the same spectrum \cite{heiss-weigert-discrete-2000}.  The set of kernel eigenvectors at the point $\Omega$ is the Dicke basis quantized along the axis $\mathbf{n}$ pointing to $\Omega$, such that we have
\begin{equation}\label{eq:SU(2)-kernel-diagonal}
    \Delta (\Omega) = \sum_{m=-j}^j \Delta_{j,m} \ketbra{j,m;\mathbf{n}}{j,m;\mathbf{n}},
\end{equation}
with rotationally-invariant eigenvalues
\begin{equation}\label{eq:kernel_eigenvalues}
    \Delta_{j,m} = \sum_{L=0}^{2j} \frac{2L+1}{2j+1} C^{j,m}_{j,m; L, 0}
\end{equation}
where $C_{j_1, m_1;j_2, m_2}^{J,M}$ are Clebsch-Gordan coefficients. In particular, at the North pole ($\Omega=0$) the kernel is diagonal in the standard Dicke basis and its matrix elements are
\begin{equation}
    [\Delta(0)]_{mn} = \langle j,m | \Delta(0) | j,n \rangle = \Delta_{j,m}\delta_{mn}.
\end{equation}
The kernel is guaranteed to have unit trace at all points and in all dimensions:
\begin{equation}\label{eq:kernel_eigs_unit_sum}
    \sum_{m=-j}^j \Delta_{j,m} = 1 \quad \forall\, j,
\end{equation}
and satisfies the relationship~\cite{Abgaryan2021}
\begin{equation}\label{identity2}
    \sum_{m=-j}^j \Delta_{j,m}^2 = 2j+1 \quad \forall\, j,
\end{equation}
for which we give a proof of in Appendix \ref{sec:remarkablerelation} for the sake of consistency.

Finally, we note the following observations on the set of kernel eigenvalues \eqref{eq:kernel_eigenvalues}:
\begin{equation}\label{eq:kernel_eigenvalue_assumption}
\begin{split}
    & |\Delta_{j,m}| > |\Delta_{j,m-1}| \neq 0, \\[2pt]
    & \sgn(\Delta_{j,k}) = (-1)^{j-k}
\end{split}
\end{equation}
for all $m \in \{-j+1,...,j \}$. That is, as $m$ ranges from $j$ to $-j$ the eigenvalues alternate in sign (starting from a positive value at $m=j$) and strictly decrease in absolute value without vanishing.  Numerics support this assumption though we are not aware of any proof;  see also \cite{davis2021wigner,koczor_parity_2020} for discussions on this point.  Note this implies that the kernel has multiplicity-free eigenvalues for all finite spin.  This is in contrast to the Wigner function on $\mathbb{R}^2$, which has a highly degenerate kernel (i.e., it acts on an infinite-dimensional Hilbert space but only has two eigenvalues) \cite{royer_parity_1977}.  Only some of our results depend on \eqref{eq:kernel_eigenvalue_assumption}, and we will highlight when this is the case.

In what follows we use the vector notation $\boldsymbol{\lambda}$ for the spectrum $(\lambda_0,\lambda_1,\ldots,\lambda_{2j})$ of a density operator $\rho$, and likewise $\boldsymbol{\Delta}$ for the spectrum $(\Delta_{j,-j}, \Delta_{j,-j+1},...,\Delta_{j,j})$ of the kernel $\Delta$.  We also alternate between the double-subscript notation $\Delta_{j,m}$, which refers directly to Eq.~\eqref{eq:kernel_eigenvalues}, and the single-subscript notation $\Delta_i$ where $i\in\{0,...,2j\}$, which denotes a vector component, similar to $\lambda_i$.

\section{Polytopes of absolutely Wigner bounded states}
\label{sec:AWPPolytope}

We present in this section our first result. We prove there exists a polytope containing all absolutely Wigner bounded (AWB) states with respect to a given lower bound, and fully characterize it.  When this bound is zero we refer to such states as absolutely Wigner positive (AWP).  We also determine a necessary and sufficient condition for a state to be inside the AWB polytope based on a majorization criterion. These results offer a strong characterization of the classicality of mixed spin states.

We start with the following definition of AWB states:
\begin{definition}
A spin-$j$ state $\rho$ is absolutely Wigner bounded (AWB) with respect to $W_\mathrm{min}$ if the Wigner function of each state unitarily connected to $\rho$ is lower bounded by $W_\mathrm{min}$.  That is, if
\begin{equation}
\begin{split}
    W_{U\rho U^\dagger}(\Omega) \geq W_\mathrm{min}
\end{split}
\qquad
\begin{split}
    &\forall \,\, \Omega \in S^2 \\
    &\forall \,\, U \in \mathrm{SU}(2j+1).
\end{split}
\end{equation}
When $W_\mathrm{min} = 0$ we refer to such states as absolutely Wigner positive (AWP).  Hence, an AWP state has only non-negative Wigner function states in its unitary orbit.
\end{definition}

\subsection{Full set of AWB states}
\label{Sec:AWBpolytopes}

The following proposition is an extension and alternative derivation of a result on absolute positivity obtained in~\cite{Abgaryan2020,Abgaryan2021}. It gives a complete characterization of the set of states whose unitary orbit contains only states whose Wigner function is larger than a specified constant value, and is valid for any spin quantum number $j$.

\begin{proposition}
\label{AWP_Theorem}
Let $\boldsymbol{\Delta}^\uparrow$ denote the vector of kernel eigenvalues sorted into increasing order, and let
\begin{equation}
    W_\mathrm{min}\in [ \Delta^\uparrow_0, \tfrac{1}{2j+1} ].
\end{equation}
Then a spin state $\rho$ has in its unitary orbit only states whose Wigner function satisfies $W(\Omega)\geq W_\mathrm{min}\;\forall \,\Omega$ iff its decreasingly ordered eigenvalues $\boldsymbol{\lambda}^\downarrow$ satisfy the following inequality
\begin{equation}\label{eq:THM}
\sum_{i=0}^{2j}\lambda_{i}^\downarrow\Delta_i^\uparrow\geq W_\mathrm{min}.
\end{equation}
\end{proposition}
\textit{Remark.}  While not necessary for the proof to hold, note that according to Eq.~\eqref{eq:kernel_eigenvalue_assumption} the sorted kernel eigenspectrum becomes $\boldsymbol{\Delta}^\uparrow=(\Delta_{j,j-1}, \Delta_{j,j-3},...,\Delta_{j,-j},...,\Delta_{j,j-2},\Delta_{j,j})$ and so $W_\mathrm{min} \in [\Delta_{j,j-1}, \frac{1}{2j+1}]$.  The upper bound comes from Eq.\ \eqref{eq:normalization}, which implies that any Wigner function with $W_\mathrm{min} > 1/(2j+1)$ would not be normalized.  Furthermore, for $W_\mathrm{min} = 0$, this proposition provides a characterisation of the set of AWP states, as previously found in a more abstract and general setting in~\cite{Abgaryan2020, Abgaryan2021}.

\begin{proof}
Consider a general spin state $\rho$. We are first looking for a necessary condition for any element $U\rho U^{\dagger}$ of the unitary orbit of $\rho$ to have a Wigner function $W(\Omega)\geq W_\mathrm{min}$ at any point $\Omega\in S^2$. Since the unitary transformation applied to $\rho$ may correspond, in a particular case, to an SU(2) rotation, the value of the Wigner function of $\rho$ at any point $\Omega$ corresponds to the value of the Wigner function at $\Omega=0$ of an element in its unitary orbit (the rotated version of $\rho$). But since we are considering the full unitary orbit, i.e., all possible $U$'s, we can set the Wigner function argument to $\Omega=0$ via the following reasoning. The state $\rho$ can always be diagonalized by a unitary matrix $M$, i.e., $M\rho M^\dagger=\Lambda$ with $\Lambda=\mathrm{diag}(\lambda_0,...,\lambda_{2j})$ a diagonal positive semi-definite matrix. The Wigner function at $\Omega=0$ of $U\rho U^{\dagger}$ is then given by
\begin{equation*}
\begin{aligned}
W_{U\rho U^{\dagger}}(0) &{}={} \mathrm{Tr}\left[U\rho U^{\dagger}\Delta(0)\right] \\
      &{}={} \mathrm{Tr}\left[U M^\dagger \Lambda M U^{\dagger}\Delta(0)\right].
\end{aligned}
\end{equation*}
By defining the unitary matrix $V=U M^\dagger$ and calculating the trace in the Dicke basis, we obtain (where we drop the Wigner function argument in the following)
\begin{equation*}
\begin{aligned}
W_{U\rho U^{\dagger}} &{}= \mathrm{Tr}\left[V \Lambda V^{\dagger}\Delta(0)\right]\\
 &{}= \sum_{p,q,k,l=0}^{2j}V_{pq}\lambda_{q}\delta_{qk}V_{lk}^{*}\Delta_{l}\delta_{lp}\\
 &{}=\sum_{q,p=0}^{2j}\lambda_{q}\left|V_{qp}\right|^{2}\Delta_{p}.
\end{aligned}
\end{equation*}
The positive numbers $|V_{qp}|^{2}$ in the previous equation define the entries of a unistochastic (hence also doubly stochastic) matrix of dimension $(2j+1)\times(2j+1)$ which we denote by $X$,
\begin{equation}
    X_{qp}=\left|V_{qp}\right|^{2}.
\end{equation}
By the Birkhoff-von Neumann theorem, we know that $X$ can be expressed as a convex combination of permutation matrices $P_{k}$,
\begin{equation}
    X=\sum_{k=1}^{N_p}c_{k}P_{k},
\end{equation}
where $N_p=(2j+1)!$ is the total number of permutations $\pi_{k} \in S_{2j+1}$ with $S_{2j+1}$ the symmetric group over $2j+1$ symbols, 
\begin{equation}
    c_{k}\geq0 \quad\forall\, k\quad \mathrm{and} \quad \sum_{k=1}^{N_p}c_{k}=1.
\end{equation}
Consequently, we have
\begin{equation*}
\begin{aligned}
    W_{U\rho U^{\dagger}} &{}= \sum_{p,q=0}^{2j}\lambda_{p}X_{pq}\Delta_{q}\\
 &{}= \sum_{k=1}^{N_p}c_{k}\sum_{p,q=0}^{2j}\lambda_{p}\left[P_{k}\right]_{pq}\Delta_{q}\\
 &{}= \sum_{k=1}^{N_p}c_{k}\sum_{p=0}^{2j}\lambda_{p}\Delta_{\pi_{k}(p)}
\end{aligned}
\end{equation*}
For a state $\rho$ whose eigenspectrum $\boldsymbol{\lambda}$ satisfies the $N_p$ inequalities
\begin{equation}
\sum_{p=0}^{2j}\lambda_{p}\Delta_{\pi(p)}\geq W_\mathrm{min} \qquad\forall\, \pi\in S_{2j+1}\label{eq:AWPCondition}
\end{equation}
we then have
\begin{equation*}
W_{U\rho U^{\dagger}} = \sum_{k=1}^{N_p}c_{k}\sum_{p=0}^{2j}\lambda_{p}\Delta_{\pi_{k}(p)} \geq  W_\mathrm{min}
\end{equation*}
for any unitary $U$ and we conclude.

Conversely, a state has in its unitary orbits only states whose Wigner function satisfies $W(\Omega)\geq W_\mathrm{min}\;\forall \,\Omega$ if
\begin{equation}
    W_{U\rho U^{\dagger}} = \sum_{k=1}^{N_p}c_{k}\sum_{p=0}^{2j}\lambda_{p}\Delta_{\pi_{k}(p)} \geq W_\mathrm{min} \qquad \forall\, U.
\end{equation}
In particular, the unitary matrix $U$ can correspond to any permutation matrix $P$, so that we have
\begin{equation}
        W_{P \rho P^{\dagger}} = \sum_{p=0}^{2j}\lambda_{p}\Delta_{\pi(p)} \geq W_\mathrm{min} \qquad \forall \, \pi
        \label{awp_ineq}
\end{equation}
and we conclude that the state satisfies \eqref{eq:AWPCondition}.

In fact, it is enough to consider the ordered eigenvalues $\boldsymbol{\lambda}^{\downarrow}$ so that a state is AWB iff it verifies the most stringent inequality \eqref{awp_ineq} amongst all possible permutations
\begin{equation}\label{eq:ordered_awp_ineq}
    \boldsymbol{\lambda}^\downarrow \boldsymbol{\cdot} \boldsymbol{\Delta}^\uparrow = \sum_{p=0}^{2j}\lambda_{p}^\downarrow \Delta_{p}^\uparrow \geq W_\mathrm{min}
\end{equation}
with the ordered eigenvalues of the kernel $\boldsymbol{\Delta}^\uparrow$.
\end{proof}

The proof provided for Proposition 1 can in fact be reproduced for any quasiprobability distribution $\mathcal{W}$ defined on the spherical phase space $S^2$ as the expectation value of a specific kernel operator $\tilde{\Delta}(\Omega)$ in a quantum state $\rho$; that is, via $\mathcal{W}_\rho(\Omega) = \mathrm{Tr}\left[\rho \tilde{\Delta}(\Omega)\right]$, see also Refs.~\cite{Abgaryan2020,Abgaryan2021} for other generalizations. A polytope in the simplex of states will describe the absolute positivity of each quasiprobability distribution and its vertices will be determined by the eigenspectrum of the defining kernel. A family of such (normalized) distributions is obtained from the $s$-parametrized Stratonovich-Weyl kernel (see, e.g., Refs.~\cite{1981Agarwal,varilly_moyal_1989,Brif1998})
 \begin{equation}\label{sSWkernel}
\Delta^{(s)}(\Omega) = \sqrt{\frac{4\pi}{2j+1}}\sum_{L,M}\left(C_{j j, L 0}^{j j}\right)^{-s} Y_{LM}^{*}(\Omega)T_{LM}
\end{equation}
with $s\in [-1,1]$. For $s=0$, it reduces to the Wigner kernel given in Eq.~\eqref{Wignerkernel}.

As negative values of the Wigner function are generally considered to indicate non-classicality, the value $W_\mathrm{min}=0$ plays a special role. Nevertheless, since Proposition 1 holds for any $W_{\mathrm{min}}\in [ \min\{\Delta_i\},\frac{1}{2j+1}]$ the corresponding sets of states also form polytopes, which become larger as $W_{\mathrm{min}}$ becomes more negative, culminating in the entire simplex when $W_{\mathrm{min}}$ is the smallest kernel eigenvalue $\min\{\Delta_i\}$ (which according to Eq.\ \eqref{eq:kernel_eigenvalue_assumption} is $\Delta_{j.j-1}$).  There is thus a continuous transition between the one-point polytope, which represents the maximally mixed state, and the polytope containing the whole simplex.  As discussed later, Fig.\ \ref{fig:criticalpolytope} in Sec.\ \ref{sec:AWP_balls} shows a special example of this family for spin-1.

Quasiprobability distributions other than the Wigner function, such as the Husimi $Q$ function derived from the $s$-ordered Stratonovich-Weyl kernel \eqref{sSWkernel} for $s=-1$, are positive by construction, implying that the polytope for $Q_{\mathrm{min}}=0$ contains the entire simplex of state spectra. In this case it becomes especially interesting to consider lower bounds $Q_{\mathrm{min}}>0$ and study the properties of the associated polytopes.

\subsection{AWP polytopes}
\label{subsec:AWPPolytope}
Since the conditions for being AWP depend only on the eigenspectrum $\boldsymbol{\lambda}$ of a state, it is sufficient in the following to focus on diagonal states in the Dicke basis. The condition \eqref{eq:THM} for $W_{\mathrm{min}}=0$ defines a polytope of AWP states in the simplex of mixed spin states. Indeed, we start by noting that the equalities \begin{equation}
    \sum_{i=0}^{2j}\lambda_{i}\Delta_{\pi(i)} = 0
    \label{eq:AWP_hyperplanes}
\end{equation}
define, for all possible permutations $\pi$, $(2j+1)!$ hyperplanes in $\mathbb{R}^{2j}$. Together they delimit a particular polytope that contains all absolutely Wigner positive states. The AWP polytopes for $j=1$ and $j=3/2$ are respectively represented  in Figs.~\ref{fig:spin-1-simplex_2} and \ref{fig:spin-3/2-simplex} in a barycentric coordinate system (see Appendix \ref{sec:barycentricCoordinatesSystem} for a reminder).  

If we now restrict our attention to ordered eigenvalues $\boldsymbol{\lambda}^\downarrow$, we get a minimal polytope; see Fig.~\ref{fig:minimalpolytope} for the case of $j=1$. The full polytope is reconstructed by taking all possible permutations of the barycentric coordinates of the vertices of the minimal polytope.
\begin{figure}
    \centering
    \includegraphics[width=0.45\textwidth]{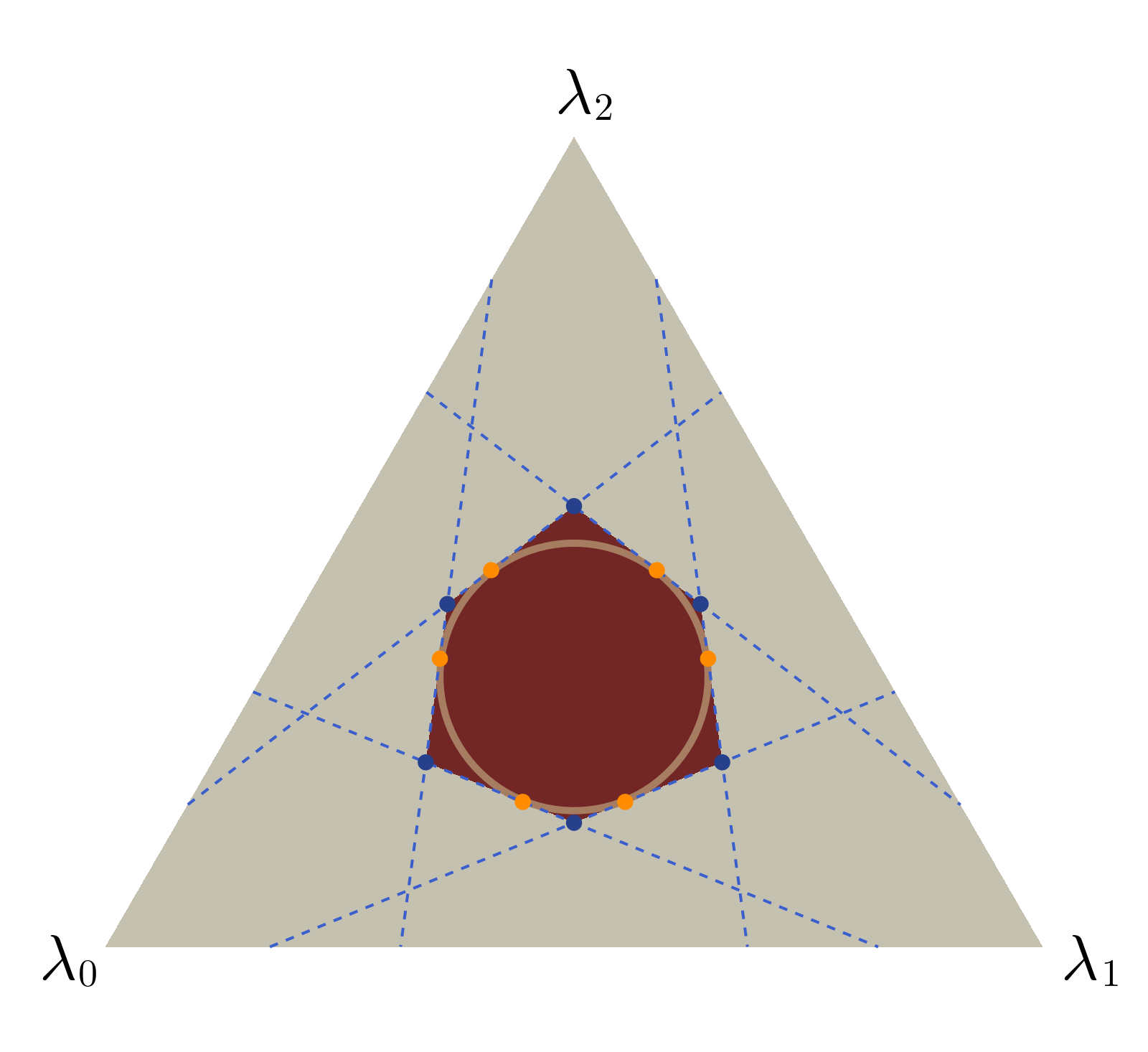}
    \caption{AWP polytope for $j=1$ displayed in the barycentric coordinate system. The AWP polytope is the area shaded in dark red with the blue dashed lines marking the hyperplanes defined by Eq.~\eqref{eq:AWP_hyperplanes}. The circle is the surface of the AWP ball (see Section \ref{sec:AWP_balls}). The orange points represent all the permutations of the spectrum \eqref{rhostar}. The gray triangle corresponds to the full simplex of spin-$1$ states with spectrum $\boldsymbol{\lambda}=(\lambda_0,\lambda_1,\lambda_2)$.}
    \label{fig:spin-1-simplex_2}
\end{figure}
These vertices can be found as follows. In general we need $2j+1$ independent conditions on the vector $(\lambda_0^{\downarrow},\lambda_1^{\downarrow},\ldots,\lambda_{2j}^{\downarrow})$ to uniquely define (the unitary orbit of) a state $\rho$.  One of them is given by the normalization condition $\sum_{i=0}^{2j}\lambda_{i}^{\downarrow} = 1$. The others correspond to the fact that a vertex of the AWP polytope is the intersection of $2j$ hyperplanes each specified by an equation of the form \eqref{eq:AWP_hyperplanes}. One of them is
\begin{equation}\label{eq:ordered_awp_eq}
\sum_{i=0}^{2j}\lambda_{i}^{\downarrow}\Delta_{i}^{\uparrow} = 0.
\end{equation}
\begin{figure}
    \centering
\includegraphics[width=0.45\textwidth]{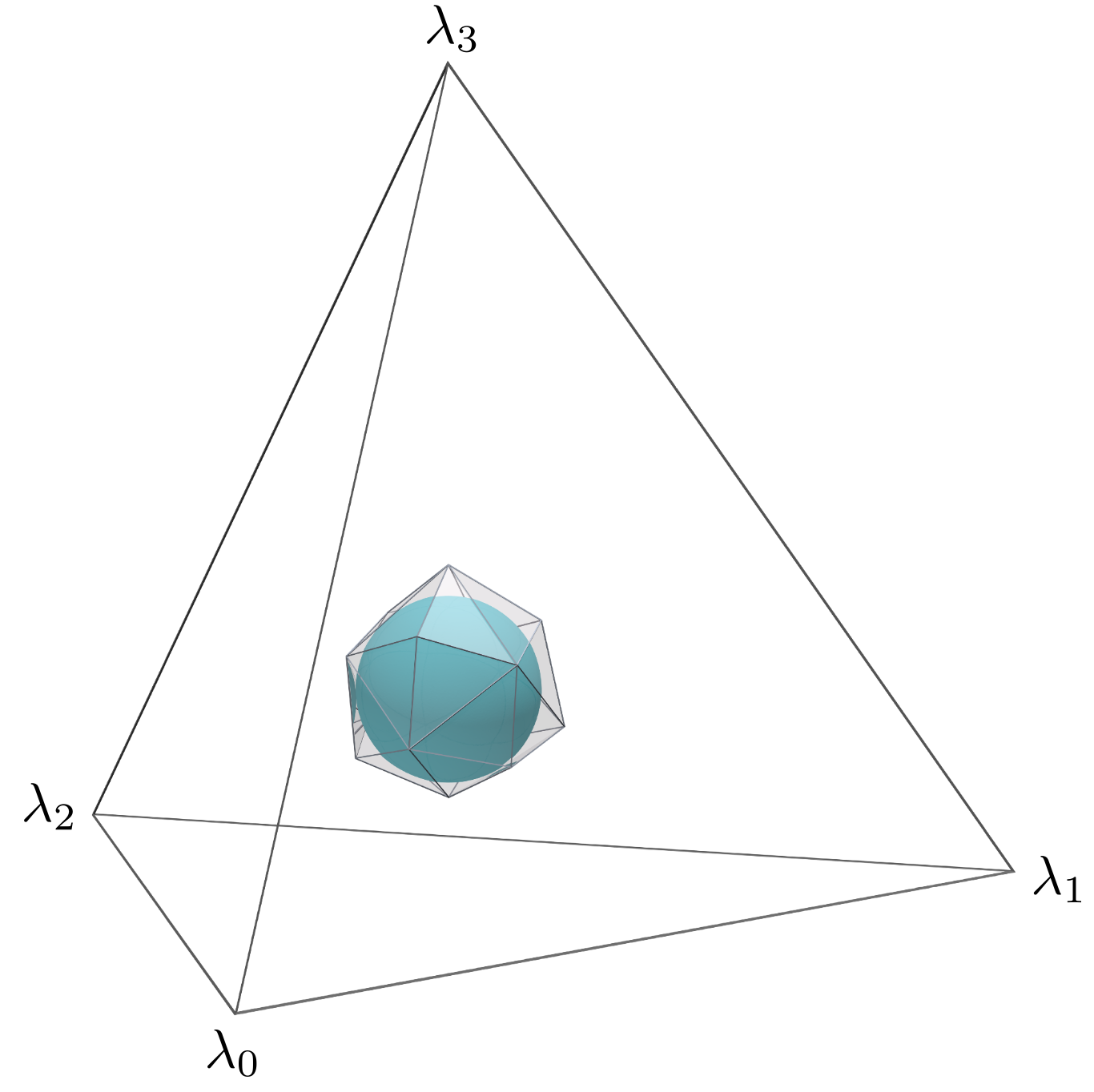}\\
\includegraphics[width=0.45\textwidth]{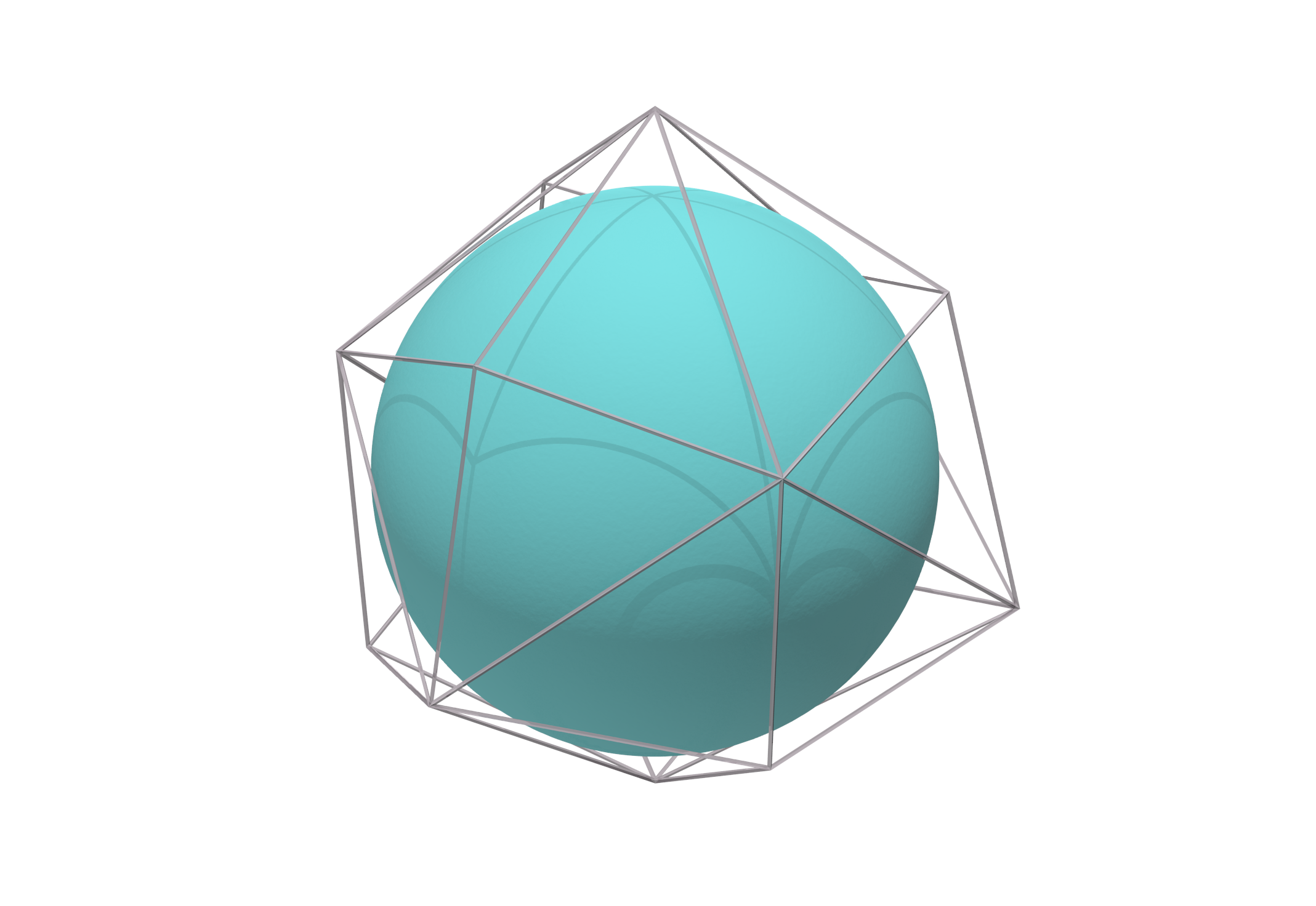}
    \caption{The AWP polytope for $j=3/2$ in the barycentric coordinate system (top). The grey rods (shown in the enlarged polytope at the bottom) are the edges of the AWP polytope and the blue sphere is its largest inner ball, with radius $r_{\mathrm{in}}^{\mathrm{AWP}}=1/(2\sqrt{15})$.}
    \label{fig:spin-3/2-simplex}
\end{figure}\begin{figure}
    \centering
    \includegraphics[width=0.45\textwidth]{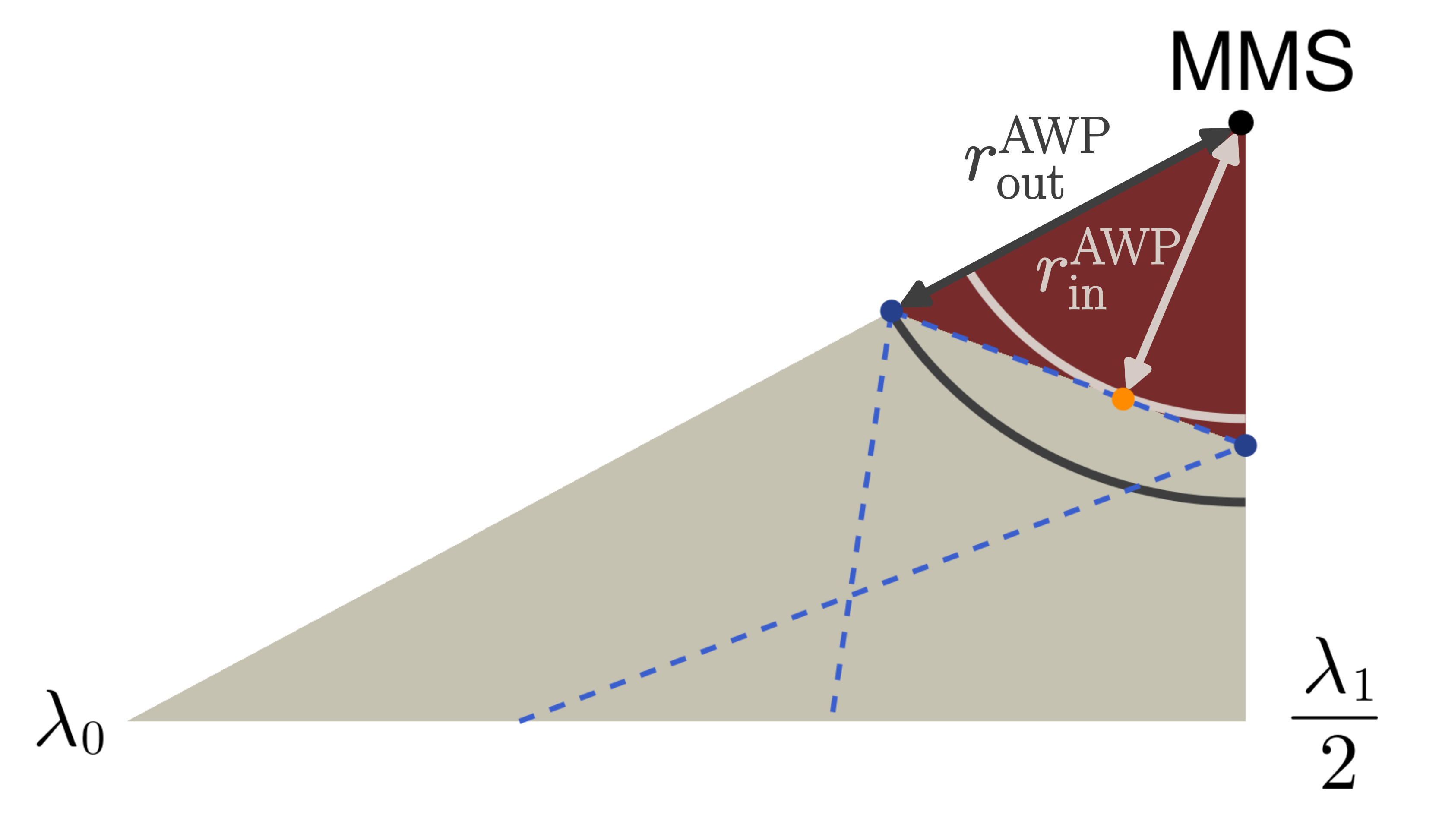}
    \caption{AWP minimal polytope for $j=1$ in the barycentric coordinate system. The structure is similar to Fig.~\ref{fig:spin-1-simplex_2} but we only draw the part where the eigenvalues of the state are ordered in descending value. The dark point corresponds to the maximally mixed state (MMS). The inner and outer AWP balls radii, $r_{\mathrm{in}}^{\mathrm{AWP}}$ and $r_{\mathrm{out}}^{\mathrm{AWP}}$, are shown.}
    \label{fig:minimalpolytope}
\end{figure}
Let us focus on the remaining $2j-1$. For simplicity, consider a transposition $\pi=(p,q)$ with $p<q$. This is the permutation whose only non-trivial action is $\pi(p) = q$ and $\pi(q)=p$. The condition \eqref{eq:AWP_hyperplanes} for this transposition becomes
\begin{align}\label{condtransposition}
& \lambda_{p}^{\downarrow}\Delta_{q}^{\uparrow}+\lambda_{q}^{\downarrow}\Delta_{p}^{\uparrow}+
        \sum_{\substack{i=0\\i\neq p,q}}^{2j}\lambda_{i}^{\downarrow}\Delta_{i}^{\uparrow} = 0 \nonumber\\[2pt]
 \Leftrightarrow\quad & \lambda_{p}^{\downarrow}(\Delta_{q}^{\uparrow}-\Delta_{p}^{\uparrow})+\lambda_{q}^{\downarrow}(\Delta_{p}^{\uparrow}-\Delta_{q}^{\uparrow})=0,
\end{align}
where the second line comes from applying the constraint \eqref{eq:ordered_awp_eq}.  As all of the kernel eigenvalues are different by assumption \eqref{eq:kernel_eigenvalue_assumption}, Eq.~\eqref{condtransposition} is satisfied iff $\lambda_{p}^{\downarrow}=\lambda_{q}^{\downarrow}$.  And because the eigenvalues are ordered this also implies $\lambda_{k}^{\downarrow}=\lambda_{p}^{\downarrow}$ for all $k$ between $p$ and $q$. The only forbidden transposition is $(0,2j)$ because it would give the maximally mixed state (MMS). Hence a given transposition $(p,q)$ will correspond to a set of $q-p$ conditions $ \lambda_{l} = \lambda_{l+1}$ for $l=p,\ldots,q-1$.
Therefore, as any permutation is a composition of transpositions, the $2j-1$ conditions that follow from \eqref{eq:AWP_hyperplanes} reduce to a set of $2j-1$ nearest-neighbour eigenvalue equalities taken from
\begin{equation}\label{eq:set_of_NN_constraints}
    \mathcal{E}=\left(\lambda_{0}^\downarrow = \lambda_{1}^\downarrow, \lambda_{1}^\downarrow=\lambda_{2}^\downarrow, ... , \lambda_{2j-1}^\downarrow=\lambda_{2j}^\downarrow\right).
\end{equation}
Since we need $2j-1$ conditions, we can draw $2j-1$ equalities from $\mathcal{E}$ in order to obtain a vertex. This method gives $\binom{2j}{2j-1}=2j$ different draws and so we get $2j$ vertices for the minimal polytope.  Geometrically, Eq.\ \eqref{eq:set_of_NN_constraints} can also be seen as the set of non-trivial hyperplanes defining the minimal polytope, and the $\binom{2j}{2j-1}=2j$ draws correspond to the different intersections of the hyperplane \eqref{eq:ordered_awp_eq} with the 1-dimensional faces (i.e., edges) of the minimal polytope; see Fig.\ \ref{fig:minimalpolytope} for an example.  The full set of hyperplanes defining the minimal polytope is \eqref{eq:set_of_NN_constraints} supplemented with $\lambda_{2j}^\downarrow = 0$ and the normalization condition.

As explained previously, all other vertices of the full polytope are obtained by permuting the coordinates of the vertices of the minimal polytope. In Appendix~\ref{sec:polytope_coordinates}, we give the barycentric coordinates of the vertices of the minimal polytope up to $j=2$. The entirety of the preceding discussion of the AWP polytope vertices naturally extends to the AWB polytope vertices for which we must replace $0$ by $W_{\text{min}}$ in the right-hand side of the equality \eqref{eq:AWP_hyperplanes}.
However, for negative values of $W_{\text{min}}$, the polytope may extend beyond the simplex and some vertices will have negative-valued components, resulting in unphysical states.

A peculiar characteristic of the AWP polytope is that each point on its surface has a state in its orbit   satisfying $W(0)=0$. Indeed, for an eigenspectrum $\boldsymbol{\lambda}$ that satisfies \eqref{eq:AWP_hyperplanes} for a given permutation $\pi$, the diagonal state $\rho$ in the Dicke basis with $\rho_{ii}=\lambda_{\pi^{-1}(i)}$ satisfies
\begin{equation}
W(0)=\sum_{i=0}^{2j}\lambda_{i}\Delta_{i} = 0
\end{equation}
and is in the unitary orbit of $\boldsymbol{\lambda}$. Following the same reasoning, in the interior of the AWP polytope, there is no state with a zero-valued Wigner function.

\subsection{Majorization condition}
\label{Sec:Majorization}

Here we find a condition equivalent to \eqref{eq:THM} for a state to be AWB based on its majorization by a mixture of the vertices of the minimal polytope.

\begin{definition}
For two vectors $\mathbf{u}$ and $\mathbf{v}$ of the same length $n$, we say that $\mathbf{u}$ majorizes $\mathbf{v}$, denoted $\mathbf{u}\succ\mathbf{v}$, iff
\begin{equation}
    \sum_{k=1}^{l}u_{k}^{\downarrow} \geq \sum_{k=1}^{l}v_{k}^{\downarrow}
\end{equation}
for $l<n$, with $\sum_{k=1}^{n}u_{k}=\sum_{k=1}^{n}v_{k}$ and $\mathbf{u}^{\downarrow}$ denoting the vector $\mathbf{u}$ with components sorted in decreasing order.
\end{definition}

\begin{proposition}
\label{AWP_Majorization}
A state $\rho$ is AWB iff its eigenvalues $\boldsymbol{\lambda}$ are majorized by a convex combination of the ordered vertices $\{\boldsymbol{\lambda}^{\downarrow}_{\mathrm{v}_k}\}$ of the corresponding AWB polytope, i.e., $\exists\,\mathbf{c}\in\mathbb{R}_{+}^{2j}$ such that
\begin{equation}\label{eq:majorizationcondition}
    \boldsymbol{\lambda}\prec\sum_{k=1}^{2j}c_{k}\boldsymbol{\lambda}^{\downarrow}_{\mathrm{v}_k}
\end{equation}
with $\sum_{k=1}^{2j}c_{k}=1$.
\end{proposition}

\begin{proof}
If $\boldsymbol{\lambda}$ is AWB then it can be expressed as a mixture of the vertices of the AWB polytope
\begin{equation}
	\boldsymbol{\lambda} = \sum_{k}c_{k}\boldsymbol{\lambda}_{\text{v}_k}
\end{equation}
and the majorization \eqref{eq:majorizationcondition} follows.

Conversely, it is known from the Schur-Horn theorem that $\mathbf{x}\succ\mathbf{y}$ iff $\mathbf{y}$ is in the convex hull of the vectors obtained by permuting the elements of $\mathbf{x}$ (i.e., the permutahedron generated by $\mathbf{x}$).  Hence, if $\boldsymbol{\lambda}$ respects \eqref{eq:majorizationcondition}, it can be expressed as a convex combination of the vertices of the AWB polytope and is therefore inside it.
\end{proof}

\subsection{Balls of absolutely Wigner bounded states}
\label{sec:AWP_balls}

In Sec.~\ref{subsec:AWPPolytope}, we have fully characterized AWB polytopes for all finite dimensions. Taking advantage of their geometry, we present here the sufficient condition for states to belong to the AWB set, based on their purity alone. Similarly, we also conjecture a necessary condition for states to be AWB. All the mathematical developments are presented in Appendix \ref{sec:App_AWBBalls}.

\begin{proposition}
\label{AWPBalls}
Denoting by $r(\rho)$ the Hilbert-Schmidt distance between a state $\rho$ and the MMS,
\begin{equation}
\label{eq:distancetoMMS}
r(\rho) = \lVert \rho - \rho_0 \rVert_{\mathrm{HS}} =  \sqrt{\mathrm{Tr}\left[\left(\rho-\rho_0\right)^2\right]},
\end{equation}
the radius of the largest inner ball of the AWB polytope associated with a $W_{\mathrm{min}}$ value such that the ball is contained within the state simplex is
\begin{equation}\label{eq:rmaxAna}
    r_{\mathrm{in}}^{W_{\mathrm{min}}} = \frac{1-(2j+1)W_{\mathrm{min}}}{2\sqrt{j(2j+1)(j+1)}}.
\end{equation}
\end{proposition}
\begin{proof}
See Appendix \ref{sec:App_AWBBalls} (Subsec.~\ref{subsec:largest_AWBBalls}).
\end{proof}

First, let's discuss this result for positive values of $W_{\mathrm{min}}$. The inner radius \eqref{eq:rmaxAna} vanishes for $W_{\mathrm{min}}=1/(2j+1)$, corresponding to the fact that only the MMS state has a Wigner function with this minimal (and constant) value. The radius then increases as $W_{\mathrm{min}}$ decreases. At $W_{\mathrm{min}}=0$, it reduces to the radius of the largest ball of AWP states,
\begin{equation}
r_{\mathrm{in}}^{\mathrm{AWP}} = \frac{1}{2\sqrt{j(2j+1)(j+1)}}.
\end{equation}
Expressed as a function of dimension $d=2j+1$ and re-scaled to generalized Bloch length, this result was also recently found in the context of SU($d$)-covariant Wigner functions (i.e., as the phase space manifold changes dramatically with each Hilbert space dimension, rather than always being the sphere) \cite{Abgaryan2021b}.  Although our bound is tight for all $j$ in the SU(2) setting (i.e., there always exist orbits infinitesimally farther away that contain Wigner-negative states), it is unknown if this bound remains tight for such SU($d$)-covariant Wigner functions for $d>2$. 

At the critical value\footnote{In the limit $j\to\infty$, as $\Delta_{j,j}\to 2$ \cite{Weigert_contracting_2000}, Eq.~\eqref{Wmincritical} tends to $-1/2$.  This is discussed in more detail in the next section.}
\begin{equation}\label{Wmincritical}
W_{\mathrm{min}}=\frac{\Delta_{j,j} - (2 j+1)}{\Delta_{j,j} (2j+1)-1}<0,
\end{equation}
the spectrum \eqref{rhostar} acquires a first zero eigenvalue, $\lambda^*_{2j}=0$. This corresponds to the situation where $\boldsymbol{\lambda}^*$ is simultaneously on the ball surface, on a face of the polytope, and on an edge of the simplex; see the orange dots in Fig.~\ref{fig:criticalpolytope}. For more negative values of $W_{\mathrm{min}}$, Eq.~\eqref{rhostar} no longer represents a physical state because $\lambda^*_{2j}$ becomes negative. In this situation, in order to determine the radius of larger balls that contain only AWB states, additional constraints must be imposed in the optimisation procedure reflecting the fact that some elements of the spectrum of $\rho$ are zero. Since the possible number of zero eigenvalues depends on $j$, we will not go further in this development. However, in the end, when there is only one non-zero eigenvalue left (equal to 1, in which case the states are pure), the most negative $W_{\mathrm{min}}$  corresponds to the smallest kernel eigenvalue $\Delta_{j.j-1}$ (according to the conjecture \eqref{eq:kernel_eigenvalue_assumption}), and the radius is the distance $r=\sqrt{2j/(2j+1)}$ from pure states to the MMS.

\begin{figure}
    \centering
    \includegraphics[width=0.45\textwidth]{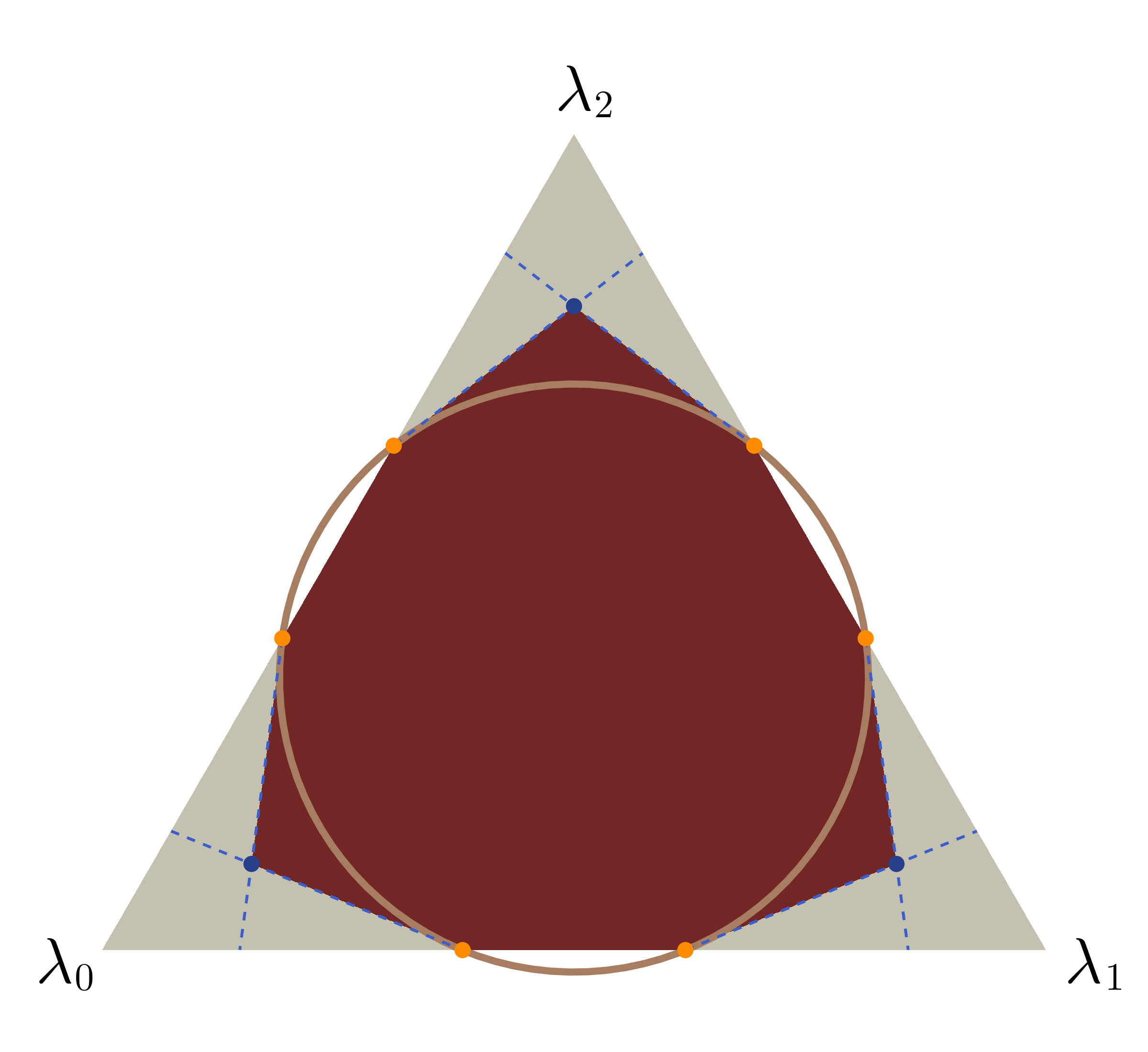}
    \caption{AWB polytope in the barycentric coordinate system for $j=1$ and $W_{\mathrm{min}}=\frac{1}{3}+\frac{2}{3} \sqrt{2} \left(\sqrt{5}-3\right)\approx -0.387$ as given by Eq.~\eqref{Wmincritical}. The structure is similar to Fig.~\ref{fig:spin-1-simplex_2} but the polytope occupies a larger portion of the state space. We omit the part of the polytope that is outside the simplex.}
    \label{fig:criticalpolytope}
\end{figure}

Finally, it should be noted that any state resulting from the permutation of the elements of $\boldsymbol{\lambda}^*$ is also on the surface of the AWB inner ball and verify a similar equality to \eqref{eq:AWP_hyperplanes} for any permutation $\pi$. Thus, considering all permutations of the elements of $\boldsymbol{\lambda}^*$ we can find all states located where the AWB polytope is tangent to the AWB inner ball, as shown in Fig.\ \ref{fig:spin-1-simplex_2} for $j=1$ and $W_{\mathrm{min}}=0$.

Proposition~\ref{AWPBalls} leads to a sufficient condition for being AWB, that is, $r\leq r_{\mathrm{in}}^{W_{\mathrm{min}}}$. A necessary condition can be obtained on the basis of the smallest outer ball containing the AWB polytope. We formulate the following conjecture for its radius.
\begin{conjecture}
\label{AWPouterBalls}
The radius of the smallest outer ball of the AWB polytope associated with a $W_{\mathrm{min}}$ value is
\begin{equation}\label{eq:r_out_conjecture}
    r_{\mathrm{out}}^{W_\mathrm{min}} = \sqrt{\frac{2j}{2j+1}} \left\lvert \frac{W_{\mathrm{min}}(2j+1) - 1}{\Delta_{j,j-1}(2j+1) - 1} \right\rvert,
\end{equation}
where $W_{\mathrm{min}} \in [ \Delta_{j,j-1}, \frac{1}{2j+1} ]$.
\end{conjecture}
A detailed argument leading to this conjecture, which interestingly involves the $W$ state of many-body entanglement \cite{Dur_LOCC_2000}, is presented in Appendix \ref{sec:App_AWBBalls} (Subsec.~\ref{smallestouterball}). Numerics strongly support the conjecture.

\subsection{Infinite spin limit}
\label{sec:infinite-spin-limit}

An important structural relation between spin and bosonic systems is the well-known contraction from the former to the latter in the limit of infinite spin \cite{Arecchi_1972}.  Under this contraction, the following identifications between operators can be made as $j \rightarrow \infty$:
\begin{equation}
        \hat J_z \rightarrow \hat N, \quad
        \hat J_+ \rightarrow \hat a, \quad
        \hat J_- \rightarrow \hat a^\dagger,
\end{equation}
where $\hat N = \hat a^\dagger \hat a$ is the number operator and $\hat a^\dagger (\hat a)$ is the creation (annihilation) operator.  For states, the Dicke basis $\ket{j,m}$ contracts to the Fock basis $\ket{n}$ via
\begin{equation}
    \ket{j,j-n} \rightarrow \ket{n}.
\end{equation}
In particular, the collective spin-up state $\ket{j,j}$ becomes the bosonic vacuum $\ket{0}$.  Such a contraction allows one to study and compare the nature of the two different types of physical systems.  Furthermore, the spherical Wigner function used here \eqref{eq:defWignerFunction} is perfectly compatible with this contraction in the sense that it tends towards the original, planar Wigner function \cite{Weigert_contracting_2000, koczor_parity_2020}.  Here we argue that our results may be leveraged through this contraction to make novel statements about Wigner negativity in the original Wigner function of general mixed states.

While the description of the AWB polytopes becomes increasingly complicated as dimension increases, the two Hilbert-Schmidt balls are well-behaved.  In the infinite-spin limit, the inner ball vanishes for all cutoffs $W_{\mathrm{min}}$ but the outer ball does not.  Hence for any cutoff $W_{\mathrm{min}} \in [-2, 0]$ there exists a ball in the infinite-dimensional Hilbert space centred on the zero operator, seen as the maximally mixed state\footnote{Of course strictly speaking the zero operator is not a quantum state but it can be seen as the limit of a sequence of thermal states with increasing temperature.  It can also be seen as the limit of finite-dimensional maximally mixed states, $\lim_{d\rightarrow\infty} \mathbb{1}_d/d$.  Either way, in the limit it remains an element of the Hilbert space of operators and so we may consider distances from it.}, that may contain AWB states.  From \eqref{eq:r_out_conjecture} we see that the radius of this ball is $|W_{\mathrm{min}}|/2$ because $\Delta_{j,j-1} \rightarrow -2$ in the contraction \cite{Weigert_contracting_2000}.  In the particular case of $W_{\mathrm{min}}=0$, the outer bosonic ball vanishes.  This implies that bosonic AWP states cannot exist.

On the other hand, if $W_{\mathrm{min}} < 0$, then there exists a bosonic ball (i.e., some subset of Hilbert space) such that any state outside of it is guaranteed to have a Wigner function that can be forced to dip below $W_{\mathrm{min}}$ under some appropriate unitary action.  This is an interesting result on how state mixedness relates to Wigner negativity in the bosonic case, which is notably complicated \cite{Narcowich_spectrum, Gracia_mixed_1988, Brocker_Werner_1995, Mandilara_extending_hudson_2010}.  The vanishing of the inner ball does not imply the lack of bosonic AWB states, it just implies there is no sufficient condition based only on purity.  Interesting further work could be to prove the existence (or lack thereof) of bosonic AWB states.

\section{Comparison between Wigner and Glauber-Sudarshan positivity}
\label{sec:entanglement}

Another common quasi-probability distribution studied in the context of single spins is the Glauber-Sudarshan function, defined through the equality
\begin{equation}\label{Pfunc}
\rho=\frac{2j+1}{4 \pi} \int P_{\rho}(\Omega)\,\ket{\Omega}\bra{\Omega}\, d \Omega.
\end{equation}
The Glauber-Sudarshan function is not unique because it is possible to add high-order spherical harmonics to it (those with $L>2j$) while maintaining equality \eqref{Pfunc}\footnote{A similar freedom exists in the definition of the Wigner function, and has been exploited in \cite{2022Fleischhauer} to eliminate the negativity of the Wigner function of a qubit. Here, however, we define a unique Wigner function \eqref{eq:defWignerFunction} by requiring that it tends to the Wigner function for continuous variables systems when $j\to\infty$.}. Negative values of all possible Glauber-Sudarshan functions representing the same state can be interpreted as the presence of entanglement within the multi-qubit realization of the system \cite{Bohnet-Waldraff-PPT_2016}. In other words, a general state $\rho$ of a single spin-$j$ system admits a positive Glauber-Sudarshan function if and only if the many-body realization is separable (necessarily over symmetric states).  This follows from the definition \eqref{Pfunc} of the Glauber-Sudarshan function as the expansion coefficients of a state $\rho$ in the spin coherent state projector basis, and the fact that spin coherent states are the only pure product states available when the qubits are indistinguishable.

States that admit a positive Glauber-Sudarshan function after any global unitary transformation are called \textit{absolutely classical} spin states \cite{Bohnet-Waldraff2017absolutely} or \textit{symmetric absolutely separable (SAS)} states \cite{SerranoEnsstiga2023}.  In this section we focus entirely on the case of $W_\mathrm{min}=0$ because negative values of the Wigner function are generally used as a witness of non-classicality and compare the AWP polytopes to the known results on SAS states.  In the context of single spins, the set of SAS states is only completely characterized for spin-1/2 and spin-1. We also show that the Wigner negative volume \eqref{eq:WignerNeg} of a positive Glauber-Sudarshan function is upper-bounded by the Wigner negative volume of a spin coherent state, see Subsec.~\ref{SubSec:BoundWignerNeg}.

\subsection{Spin-1/2}

In the familiar case of a single qubit state $\rho$, the spectrum $(\lambda, 1 - \lambda)$ is characterized by one number $\lambda$.  The kernel eigenvalues, Eq.\ \eqref{eq:kernel_eigenvalues}, are
\begin{equation}
    \Delta_{0} = \frac{1}{2}(1 - \sqrt{3}), \quad \Delta_{1} = \frac{1}{2}(1 + \sqrt{3}) = 1 - \Delta_0.
\end{equation}
Letting $\lambda \geq \frac{1}{2}$ denote the larger of the two eigenvalues, the strong ordered form \eqref{eq:ordered_awp_ineq} becomes
\begin{equation}
\begin{split}
    \lambda_0 \Delta_0 + \lambda_1 \Delta_1 &= \lambda \Delta_0 + (1-\lambda)(1 - \Delta_0) \\
   &= \lambda(2\Delta_0 - 1) + 1 - \Delta_0.
\end{split}
\end{equation}
Thus the AWP polytope is described, in the 1-dimensional projection to the $\lambda$ axis, as
\begin{equation}
    \frac{1}{2} \leq \lambda \leq \frac{1-\Delta_0}{1-2\Delta_0} = \frac{1}{2} + \frac{1}{2\sqrt{3}}.
\end{equation}
This may be equivalently expressed either in terms of purity $\gamma$ or Bloch length $|\mathbf{n}| = \sqrt{2\gamma - 1}$,
\begin{equation}
    \frac{1}{2} \leq \gamma \leq \frac{2}{3} \qquad \text{and} \qquad  |\mathbf{n}| \leq \frac{1}{\sqrt{3}}.
\end{equation}
Additionally, the distance to the maximally mixed state via Eq.\ \eqref{eq:distancetoMMS} is $r \leq  1/\sqrt{6}$, which matches with the smallest ball of AWP states derived earlier, Eq.\ \eqref{eq:rmaxAna}.  In the case of spin-1/2 this radius coincides with the largest ball containing nothing but AWP states.

Regarding absolute Glauber-Sudarshan-positivity, all qubit states are SAS.  This is a consequence of the qubit pure states being equivalent to  spin-1/2 coherent states. Thus AWP qubit states are a strict subset of SAS qubit states.

\begin{figure}
    \centering
    \includegraphics[width=0.475\textwidth]{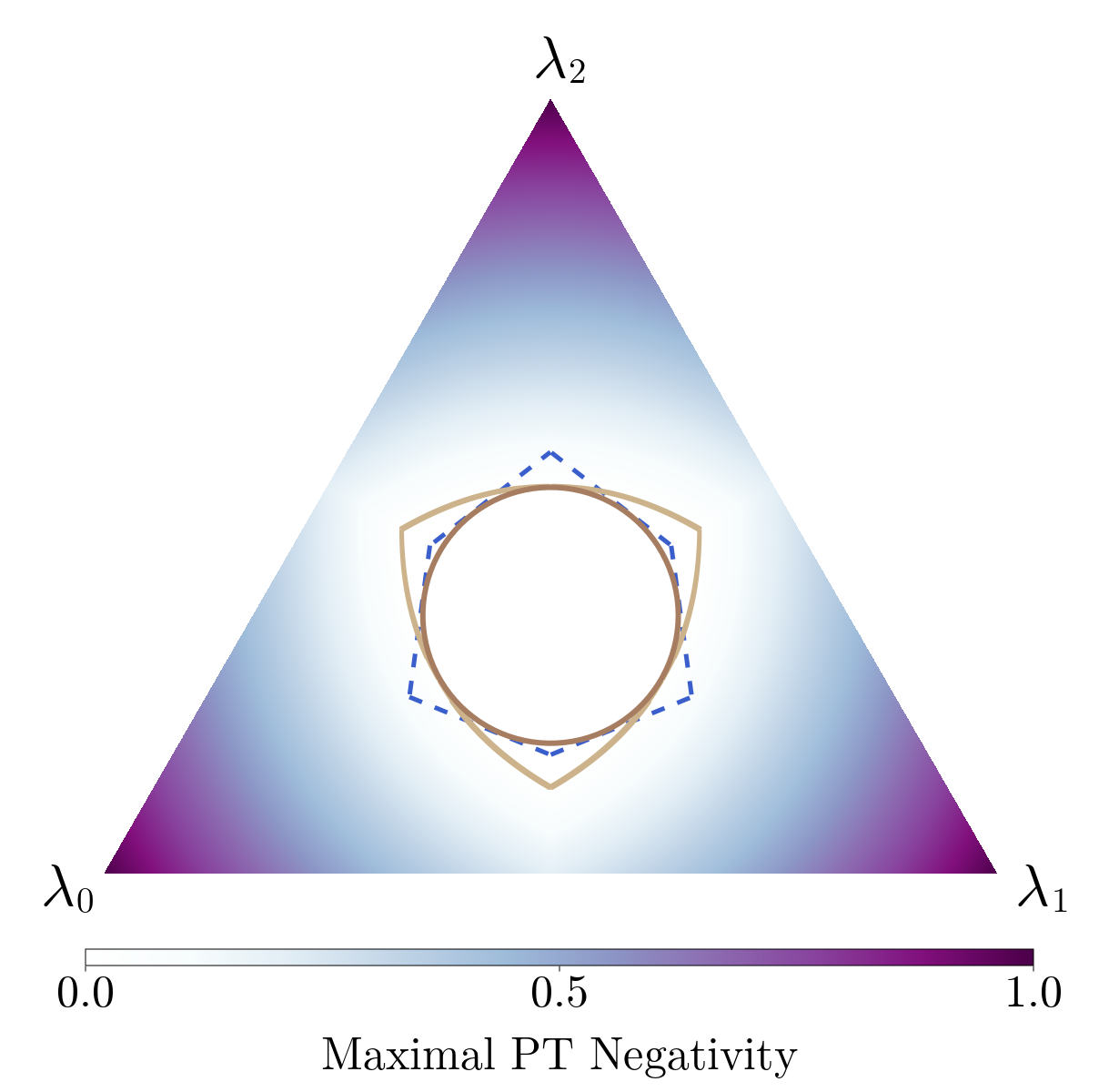}
    \caption{Maximal PT negativity over each unitary orbit in the $j=1$ simplex of state spectra.  The dashed blue line and red circle are respectively the AWP polytope and ball. The camel curve shows the boundary at which the negativity along the unitary orbit becomes non-zero.}
    \label{fig:OrbitMaximalNegativity_N=2}
\end{figure}

Furthermore, due to the invariance of negativity under rigid rotation, for a single qubit there is no distinction between a state being positive (in either the Wigner or $P$ sense) and being absolutely positive.  This means that any state with Bloch radius $|\mathbf{n}|\in (1/\sqrt{3},1]$ has a positive $P$ function but a negative Wigner function.  This is perhaps the simplest example of the fact that, unlike the planar phase space associated with optical systems, in spin systems Glauber-Sudarshan positivity does not imply Wigner positivity.

\subsection{Spin-1}
\label{sec:spin1}
For qutrits the set of AWP states and the set of SAS states are both more complicated, with neither being a strict subset of the other.  For SAS states we need the following result in \cite{SerranoEnsstiga2023}: \emph{the maximal value of the negativity, in the sense of the PPT criterion, in the unitary orbit of a two-qubit symmetric (or equivalently a spin-1) state $\rho$ with spectrum $\lambda_0\geq\lambda_1\geq\lambda_2$ is}
\begin{equation}
\label{eq:maxNeg_j1/2}
    \max\left[ 0,\sqrt{\lambda_0^2+(\lambda_1-\lambda_2)^2}-\lambda_1-\lambda_2 \right].
\end{equation}
In Fig.~\ref{fig:OrbitMaximalNegativity_N=2}, we plot the resulting maximal negativity in the $j=1$ simplex with the AWP polytope. There are clearly regions of spectra that satisfy either, both, or neither of the AWP and SAS conditions.  Thus already for spin-1 there exist states with a positive $P$ function and a negative $W$ function and vice-versa.  For $j=1$ specifically, it was also shown in \cite{SerranoEnsstiga2023} that the \emph{largest} ball of SAS states has a radius $r_{\mathrm{in}}^{P}=1/(2\sqrt{6})\approx 0.20412$, which is the same value as the radius $r_{\mathrm{in}}^{\mathrm{AWP}}=1/(2\sqrt{6})$. Hence, for $j=1$, the largest ball of AWP states coincides with the largest ball of SAS states as we can see in Fig.~\ref{fig:OrbitMaximalNegativity_N=2}.

We now illustrate the procedure described in Appendix \ref{smallestouterball} and compute the vertex states and their radii for the case of spin-$1$. The two diagonal states associated to the vertices of the minimal polytope for $j=1$ (see Fig.\ \ref{fig:minimalpolytope}) are 
\begin{align}
        \rho_{\text{v}_1} &= \omega_1 \ketbra{1,-1}{1,-1} \nonumber \\
        &\quad + \frac{1-\omega_1}{2}(\ketbra{1,0}{1,0} + \ketbra{1,1}{1,1} ), \\
        \rho_{\text{v}_2} &= \omega_2 ( \ketbra{1,-1}{1,-1} + \ketbra{1,0}{1,0} ) \nonumber \\
        &\quad + (1 - 2\omega_2)\ketbra{1,1}{1,1}
\end{align}
where the parameters $\omega_1$ and $\omega_2$ are found by solving the AWP criterion \eqref{eq:ordered_awp_eq}:
\begin{equation}
    \begin{split}
        \omega_1 &= \frac{\Delta_{1,-1} + \Delta_{1,1}}{\Delta_{1,-1} + \Delta_{1,1} - 2\Delta_{1,0}} = \frac{1}{15}(5 + \sqrt{10}),\\
    \omega_2 &= \frac{\Delta_{1,1}}{2\Delta_{1,1}-\Delta_{1,0}-\Delta_{1,-1}} = \frac{1}{6} \left(2 + \sqrt{7-3 \sqrt{5}}\right).
    \end{split}
\end{equation}
The two Hilbert-Schmidt radii \eqref{eq:distancetoMMS} of the vertex states are then
\begin{equation}
\begin{split}
    r_{\text{v}_1} &= r_{\text{out}}^\mathrm{AWP} = \frac{1}{\sqrt{15}} \approx 0.2582, \\
    r_{\text{v}_2} &= \sqrt{\frac{1}{6} \left(7-3 \sqrt{5}\right)} \approx 0.2205 .
\end{split}
\end{equation}
As conjectured, we see that $r_{\text{v}_1} = r_{\text{out}}^W$ for spin-1.

\begin{figure}[b!]
    \centering
\includegraphics[width=0.45\textwidth]{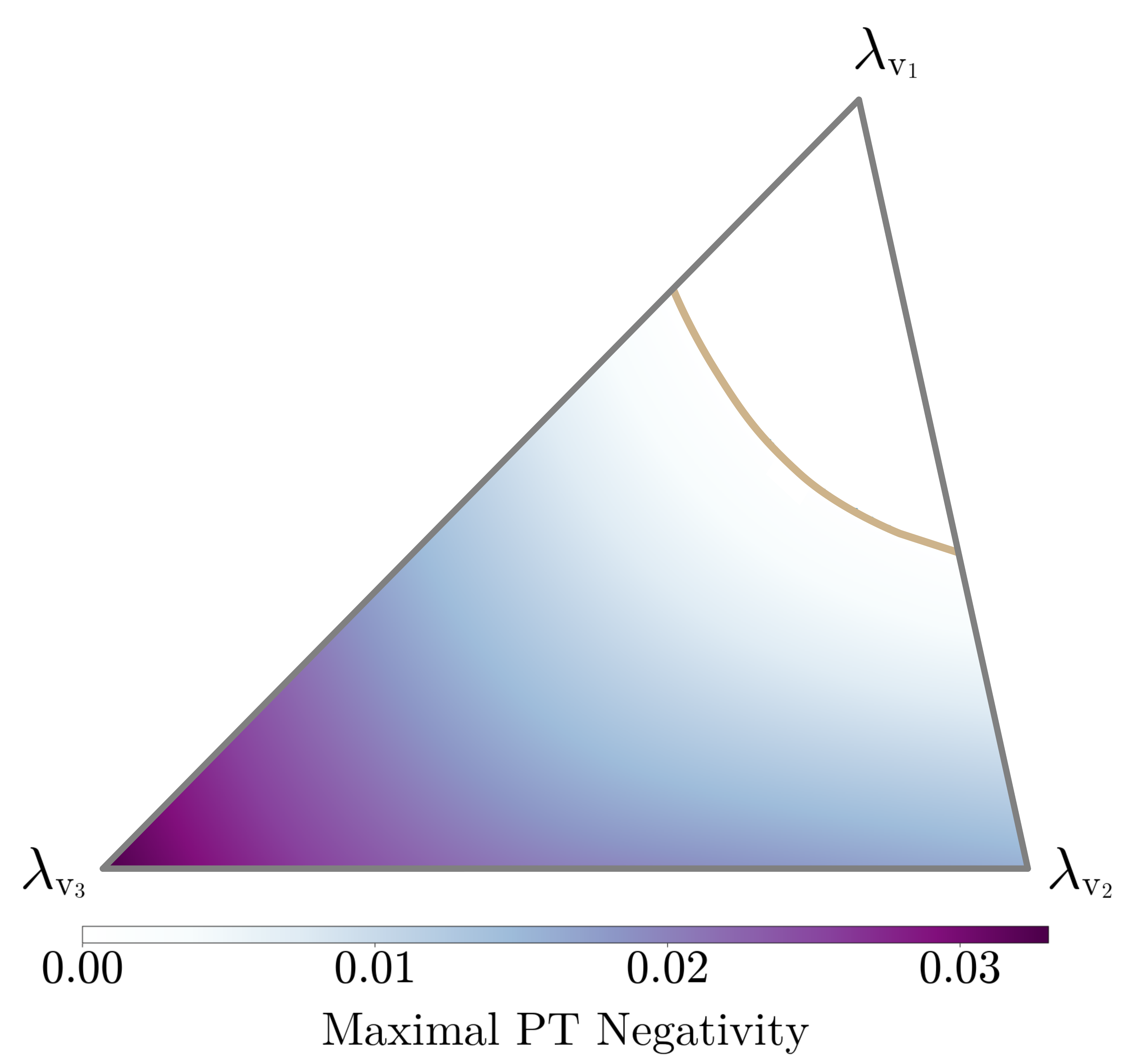}
    \caption{Maximal PT negativity over each unitary orbit on the face of the minimal $j=3/2$ AWP polytope. The camel curve shows the boundary at which the negativity along the unitary orbit becomes non-zero. The notation of the vertices corresponds to the eigenspectra given in Table \ref{table:polytopeVertices}.}
\label{fig:OrbitMaximalNegativity_N=3}
\end{figure}

\subsection{Spin-3/2}
\label{sec:spin3/2}

For spin-$3/2$, a numerical optimization (see Ref.~\cite{SerranoEnsstiga2023} for more information) yielded the maximum negativity (in the sense of the negativity of the partial transpose of the state) in the unitary orbit of the states located on a face of the polytope. The results are displayed in Fig.~\ref{fig:OrbitMaximalNegativity_N=3} where, similar to the spin-1 case, we observe both SAS and entangled states on the face of the minimal AWP polytope. A notable difference is that, for $j=3/2$, the largest ball containing only SAS states has a radius 
$r_{\mathrm{in}}^{P}=1/(2\sqrt{19})$~\cite{SerranoEnsstiga2023} which is strictly smaller than $r_{\mathrm{in}}^{\mathrm{AWP}}=1/(2\sqrt{15})$. Therefore, the SAS states on the face of the polytope are necessarily outside this ball.

\subsection{Spin-\texorpdfstring{$j>3/2$}{TEXT}}

In Fig.~\ref{fig:radiicomparison}, we compare the radius of the AWP ball (\ref{eq:rmaxAna}) with the lower bound on the radius of the ball of SAS states \cite{Bohnet-Waldraff2017absolutely}
\begin{equation}
\label{eq:boundSAS}
    r^P \equiv \frac{\left[(4j+1)\tbinom{4j}{2j}-(j+1)\right]^{-1/2}}{\sqrt{4j+2}}\leq r^P_{\mathrm{in}}.
\end{equation}
This plot suggests that the balls of AWP states can be much larger than the balls of SAS states. This is confirmed by our numerical observations that sampling the hypersurface of the polytope for $j=2$, $5/2$ and $3$ always yields states that have negative partial transpose in their unitary orbit. We also plot in Fig.~\ref{fig:radiicomparison} the conjectured radius $r_{\mathrm{out}}^\mathrm{AWP}$ of the minimal ball containing all AWP states.

Notably, the scalings of $r_{\mathrm{out}}^\mathrm{AWP}$ and $r_{\mathrm{in}}^\mathrm{AWP}$ with $j$ are different.  The scaling $r_{\mathrm{in}}^\mathrm{AWP} \propto j^{-3/2}$ follows directly from Eq.~\eqref{eq:rmaxAna}. The scaling $r_{\mathrm{out}}^\mathrm{AWP} \propto j^{-1}$ can be explained by noting that the infinite-spin limit of the SU(2) Wigner kernel is the Heisenberg-Weyl Wigner kernel, which only has the two eigenvalues $\pm 2$; see Sec.\ \ref{sec:infinite-spin-limit} and Refs.~\cite{Weigert_contracting_2000, koczor_parity_2020}.  Hence for sufficiently large $j$ we may approximate $\Delta_{j,j-1} \approx -2$.  The Laurent series of Eq.\ \eqref{eq:r_out_conjecture} with this approximation and $W_{\mathrm{min}} = 0$ has leading term $1/(4j)$, exactly matching the results shown in Fig.\ \ref{fig:radiicomparison}.

\begin{figure}
    \centering
\includegraphics[width=0.475\textwidth]{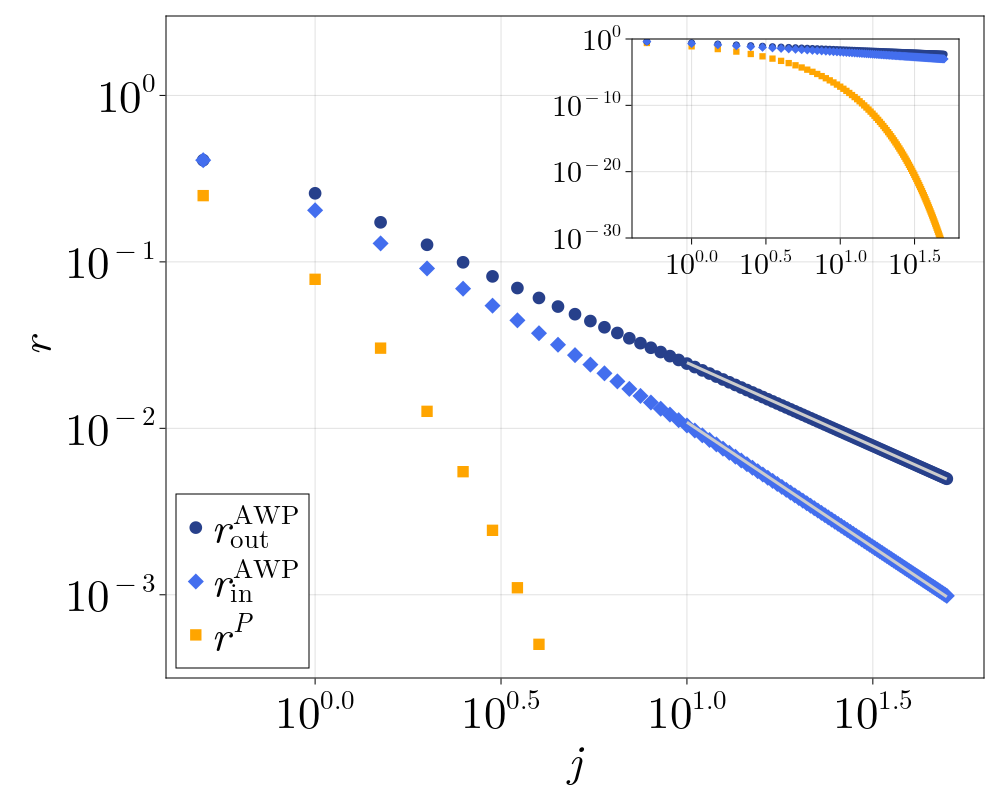}
    \caption{Comparison of the radii of the outer AWP ball (dark blue) and the inner AWP ball (blue)  and the lower bound on the SAS ball radius (orange). For $j\geq10$, we found excellent fits with $r_{\mathrm{out,fit}}^\mathrm{AWP}=0.25\times j^{-1}$ and $r_{\mathrm{in,fit}}^\mathrm{AWP}=0.336\times j^{-1.5}$. These are explained in the text.}
    \label{fig:radiicomparison}
\end{figure}

\subsection{Bound on Wigner negative volume}
\label{SubSec:BoundWignerNeg}
Here we include a related but independent result on the negative volume of spin states.
The spin-1 case showed us that there are SAS states outside the AWP polytope, i.e., with a Wigner function admitting negative values. Here, we show quite generally that the Wigner negative volume \eqref{eq:WignerNeg} of states with an everywhere positive Gluaber-Sudarshan function (in particular SAS states), denoted hereafter by $\rho_{P\geqslant0}$, is upper bounded by the Wigner negative volume of spin coherent states. Indeed, such states can always be represented as a mixture of coherent states
\begin{equation} \rho_{P\geqslant0}=\sum_{i}w_{i}\left|\alpha_{i}\right\rangle \left\langle \alpha_{i}\right|
\end{equation}
with $w_{i}\geqslant 0$ and $\sum_{i}w_{i}=1$.
Their Wigner negative volume can then be upper bounded as follows
\begin{equation} 
\begin{aligned}
    \delta(\rho_{P\geqslant0})	& =	\frac{1}{2}\int_{\Gamma}\left|W_{\rho_{P\geqslant0}}(\Omega)\right|d\mu(\Omega)-\frac{1}{2}\\
    &
    =	\frac{1}{2}\int_{\Gamma}\left|\sum_{i}w_{i}W_{\left|\alpha_{i}\right\rangle }(\Omega)\right|d\mu(\Omega)-\frac{1}{2}\\
    &
    \leqslant	\underbrace{\sum_{i}w_{i}}_{=1}\underbrace{\left(\frac{1}{2}\int_{\Gamma}\left|W_{\left|\alpha_{i}\right\rangle }(\Omega)\right|d\mu(\Omega)\right)}_{=\delta\left(|\alpha\rangle\right)+\frac{1}{2}}-\frac{1}{2}\\
    & =\delta\left(|\alpha\rangle\right)
\end{aligned} 
\end{equation}
 where $\delta\left(|\alpha\rangle\right)$ is the Wigner negative volume of a spin coherent state. Since it has been observed that such volume decreases with $j$~\cite{davis2021wigner}, the same is true for states with positive Glauber-Sudarshan function.

\section{Conclusion}
\label{sec:conclusion}
  
Given the ever-rising importance of spin quasiprobability distributions in fields like quantum information science, quantum many-body dynamics \cite{2022Fleischhauer, 2023Fleischhauer, 2021Rabl}, and quantum thermodynamics \cite{Gherardini_quasi_in_thermo_PRX_2024}, we have studied in this work the properties of the spin Wigner function of finite-dimensional quantum systems, in particular the non-classicality of the unitary orbits of mixed spin-$j$ states, highlighting important differences with infinite-dimensional systems with continuous variables. Our results shed new and interesting light on the positivity of the Wigner function in spin-$j$ systems, focusing on its relation with purity and entanglement. Our first result is Proposition~\ref{AWP_Theorem}, which gives a complete characterization for any spin quantum number $j$ of the set of absolutely Wigner bounded (AWB) states in the form of a polytope centred on the maximally mixed state in the simplex of mixed spin states.  This amounts to an extension and alternative derivation of results from \cite{Abgaryan2020, Abgaryan2021} in the setting of quantum spin.  We have studied the properties of the vertices of this polytope for different spin quantum numbers, as well as its largest inner and smallest outer Hilbert-Schmidt balls.  In particular, we have shown that the radii of the inner and outer balls scale differently as a function of $j$ (see Eqs.~\eqref{eq:rmaxAna} and \eqref{eq:r_out_conjecture} as well as Fig.~\ref{fig:radiicomparison}).  We have provided an equivalent condition for a state to be AWB based on majorization theory (Proposition~\ref{AWP_Majorization}).  We have compared our results on the positivity of the Wigner function with those on the positivity of the spherical Glauber-Sudarshan function, the latter of which can be equivalently used as a classicality criterion for spin states or a separability criterion for symmetric multiqubit states. The spin-1 and spin-3/2 cases, for which analytical results are known, were closely examined and important differences were highlighted, such as the existence of Wigner-negative absolutely separable states, and, conversely, the existence of entangled absolutely Wigner-positive states.  This novel fact represents a key distinction from the infinite-dimensional setting where a positive Glauber-Sudarshan function trivially implies a positive Wigner function \cite{Cahill_Glauber_quasi_1969, Lee_nonclassical_depth_1991}.  The infinite-spin limit of these polytopes and their Hilbert-Schmidt balls have been analyzed, and it was concluded that absolute Wigner-positivity cannot exist in infinite dimensions.  Interestingly however, our techniques cannot rule out the existence of absolutely Wigner-bounded states for non-zero cutoffs in the bosonic setting because the outer Hilbert-Schmidt ball, which represents a necessary AWB condition, does not vanish in the infinite-spin limit.  Future work is needed to investigate the possible existence of such states.

There are several other directions for future work. A notable observation drawn from our numerics is that the set of SAS states appears to shrink relative to the set of AWP states as $j$ increases, which in turn occupies a progressively smaller volume of the simplex. Further research is needed to explore this behaviour, using, e.g., the group-theoretic results of \cite{Klimov_Chumakov_2000}. A related direction could be to explore the ratio of the volume of the AWB polytopes to the volume of the full simplex; this would basically be a \textit{global indicator of classicality} like those introduced and studied in Refs.~\cite{Abbasli2020, Abgaryan2020, Abgaryan2021b} particularised to spin systems.

Another perspective, as briefly mentioned in Sec.\ \ref{sec:AWPPolytope}, is to apply the techniques presented here to other distinguished quasiprobability distributions. For example, preliminary results suggest that the absolutely Husimi bounded (AHB) polytopes have the same geometry as the simplex, but are simply reduced in size by a factor depending on $Q_{\mathrm{min}}\in[0,\tfrac{1}{2j+1}]$.  Future work could explore this further and investigate its consequences for the geometric measure of entanglement of mixed multiqubit symmetric states. Another idea is to study how these polytopes change with respect to the spherical $s$-ordering parameter (see Eq.~\eqref{sSWkernel}).

Finally, it would be intriguing to connect the lower bound on the Wigner function in the unitary orbit to the accuracy achievable in simulating the general unitary (or even dissipative) many-body quantum dynamics of spin systems efficiently using stochastic trajectories~\cite{2022Fleischhauer, 2023Fleischhauer, 2021Rabl}. Indeed, it could be expected that states with a positive lower bound would be more accurately simulated by these trajectories where the Wigner function is used as an actual probability distribution obeying a certain Fokker-Planck equation. Additionally, this bound could be linked to potential quantum advantages in applications such as parameter estimation, where the presence of negative values in quasiprobability distributions might enhance the precision of quantum sensing protocols. However, these questions go beyond the scope of this work and merit further study.

\begin{acknowledgments}
We would like to thank Yves-Eric Corbisier for his help in creating Fig.~\ref{fig:spin-3/2-simplex} with Blender~\cite{Blender}. Most of the other figures were produced with the package Makie~\cite{Makie}. We would also like to thank V.\ Abgaryan and his colleagues for their correspondence regarding Refs.\ \cite{Abgaryan2021,  Abgaryan2020, Abgaryan2021b}. This work was supported in part by the Natural Sciences and Engineering Research Council of Canada. Computational resources were provided by the Consortium des Equipements de Calcul Intensif (CECI), funded by the Fonds de la Recherche Scientifique de Belgique (F.R.S.-FNRS) under Grant No. 2.5020.11. 
\end{acknowledgments}

\appendix

\section{Proof of relation \eqref{identity2}}
\label{sec:remarkablerelation}
We show here that the eigenvalues $\Delta_{m}\equiv \Delta_{j,m}$ of the Wigner kernel \eqref{Wignerkernel} verify
\begin{equation}
    \sum_{m=-j}^j \Delta_{m}^2 = 2j+1.
\end{equation}
Using the expression \eqref{eq:kernel_eigenvalues} we get
\begin{equation}
\label{eq:sum_ev_squared}
    \begin{aligned}
    \sum_{m=-j}^j\Delta_m^2 = \sum_{m=-j}^j\sum_{L,L'=0}^{2j}&\frac{(2L+1)(2L'+1)}{(2j+1)^2}\\
    & \times C_{j,m;L,0}^{j,m}C_{j,m;L',0}^{j,m}
    \end{aligned}
\end{equation}
The Clebsh-Gordan coefficients satisfy the following relations~\cite{Varshalovich_1988}
\begin{equation}\label{Clebrule1}
    C_{a,\alpha;b,\beta}^{c,\gamma}=(-1)^{a-\alpha}\sqrt{\frac{2c+1}{2b+1}}C_{a,\alpha;c,-\gamma}^{b,-\beta}
\end{equation}
\begin{equation}\label{Clebrule2}
    \sum_{\alpha,\beta=-j}^j C_{a,\alpha;b,\beta}^{c,\gamma}C_{a,\alpha;b,\beta}^{c',\gamma'} = \delta_{cc'}\delta_{\gamma\gamma'}.
\end{equation}
Hence, by splitting the sum over $m$ in two
\begin{equation}
    \sum_{m}C_{j,m;L,0}^{j,m}C_{j,m;L',0}^{j,m} = \sum_{m_1,m_2}C_{j,m_1;L,0}^{j,m_2}C_{j,m_1;L',0}^{j,m_2}
\end{equation}
and using \eqref{Clebrule1} and \eqref{Clebrule2}, 
we get from \eqref{eq:sum_ev_squared}
\begin{equation}
    \begin{aligned}
    \sum_{m=-j}^j \Delta_{m}^2 & = \frac{1}{2j+1}\underbrace{\sum_{L=0}^{2j} 2L+1}_{=(2j+1)^2}\\
    & = 2j+1
    \end{aligned}
\end{equation}

\section{Barycentric coordinates}
\label{sec:barycentricCoordinatesSystem}

A mixed spin-$j$ state necessarily has eigenvalues $\lambda_i$ that are positive and add up to one: 
\begin{equation}
\lambda_i\geq 0,\qquad \sum_{i=0}^{2j}\lambda_i=1.
\end{equation}
This means that every state $\rho$ has its eigenvalue spectrum in the probability simplex of dimension $2j$. For example, for $j=1$, this simplex is a triangle shown in grey in Fig.~\ref{fig:barycentricCoordinatesSystem}. In geometric terms, the spectrum $(\lambda_0,\lambda_1,\lambda_2)$ defines the barycentric coordinates of a point $\boldsymbol{\lambda}$ in the simplex, as it can be considered as the centre of mass of a system of $2j$ masses placed on the vertices of the triangle.

\begin{figure}
\centering
\includegraphics[width=0.375\textwidth]{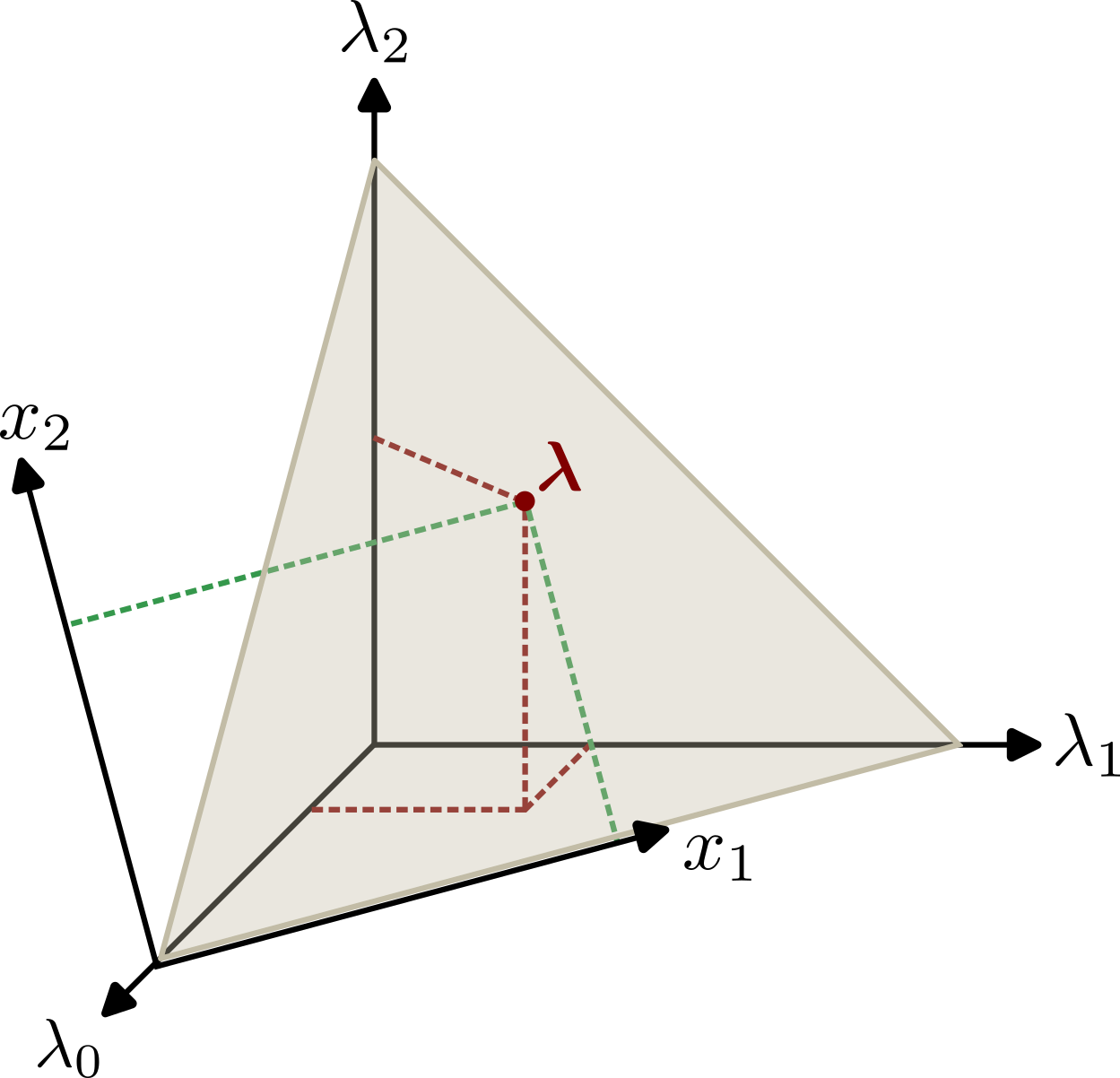}
    \caption{Barycentric and cartesian coordinate systems of spin state spectra for $j=1$. The simplex in this case is an equilateral triangle, shown here in gray. The red dot corresponds to a given spectrum and its projections onto the barycentric and Cartesian coordinate system are indicated by the red and green dashed lines respectively.}
    \label{fig:barycentricCoordinatesSystem}
\end{figure}

Let's explain how to go from the barycentric coordinate system to the Cartesian coordinate system spanning the simplex. If we denote by $\{\mathbf{r}^{(i)}: i=0,\ldots,2j\}$ the set of $2j+1$ vertices of the simplex, the Cartesian coordinates of a point $\boldsymbol{\lambda}$ are given by
\begin{equation}
    x_k = \sum_{i=0}^{2j} \lambda_i \, r_{k}^{(i)}
\end{equation}
where $r_{k}^{(i)}$ is the $k$-th Cartesian coordinate of the $i$-th vertex of the simplex. For $j=1$, the simplex is an equilateral triangle with vertices having Cartesian coordinates  $\mathbf{r}_1=(0,0)$, $\mathbf{r}_2=(1,0)$ and $\mathbf{r}_3=(1/2,\sqrt{3}/2)$. For $j=3/2$, it is a regular tetrahedron with vertices having Cartesian coordinates $\mathbf{r}_1=(0,0,0)$, $\mathbf{r}_2=(1,0,0)$, $\mathbf{r}_3=(1/2,\sqrt{3}/2,0)$ and $\mathbf{r}_4=(1/2,(2\sqrt{3})^{-1}, \sqrt{2/3})$.

\section{AWP polytope vertices for $j\leq 2$}
\label{sec:polytope_coordinates}

We give in Table~\ref{table:polytopeVertices} for $j\leq 2$ the spin state spectra associated with the vertices of the minimal AWP polytope as they can be determined as explained in Sec.~\ref{subsec:AWPPolytope}.

\begin{table}[h!]
\begin{centering}
\begin{tabular}{|c|c|c|}
\hline 
$j$ & Vertices in  barycentric coordinates \tabularnewline
\hline 
\hline 
1/2 & $\boldsymbol{\lambda}_{\mathrm{v}_1}\approx $ (0.789, 0.211)\tabularnewline
\hline 
1 & $\boldsymbol{\lambda}_{\mathrm{v}_1} \approx $ (0.423, 0.423, 0.153)\tabularnewline
 & $\boldsymbol{\lambda}_{\mathrm{v}_2} \approx $ (0.544, 0.228, 0.228)\tabularnewline
\hline 
3/2 & $\boldsymbol{\lambda}_{\mathrm{v}_1}\approx $ (0.294, 0.294, 0.294, 0.119)\tabularnewline
 & $\boldsymbol{\lambda}_{\mathrm{v}_2} \approx $ (0.33, 0.33, 0.170, 0.170)\tabularnewline
 & $\boldsymbol{\lambda}_{\mathrm{v}_3} \approx $ (0.4, 0.2, 0.2, 0.2)\tabularnewline
 \hline 
2 & $\boldsymbol{\lambda}_{\mathrm{v}_1} \approx $ (0.313, 0.172, 0.172, 0.172, 0.172)\tabularnewline
 & $\boldsymbol{\lambda}_{\mathrm{v}_2} \approx $ (0.266, 0.266, 0.156, 0.156, 0.156)\tabularnewline
  & $\boldsymbol{\lambda}_{\mathrm{v}_3} \approx $ (0.24, 0.24, 0.24, 0.14, 0.14)\tabularnewline
   & $\boldsymbol{\lambda}_{\mathrm{v}_4} \approx $ (0.226, 0.226, 0.226, 0.226, 0.097)\tabularnewline
\hline 
\end{tabular}
\caption{Barycentric coordinates (corresponding to the eigenspectrum of a mixed spin state) of the vertices of the minimal polytope of AWP states. \label{table:polytopeVertices}}
\par\end{centering}
\end{table}

\section{Inner and outer AWB balls}
\label{sec:App_AWBBalls}

\subsection{Largest ball containing only AWB states}
\label{subsec:largest_AWBBalls}
Let us first consider the radius $r_{\mathrm{in}}^{W_\mathrm{min}}$ of the largest ball centered on the MMS contained in the polytope of AWB states and find a state $\rho^*$ that is both on the surface of this ball and on a face of the polytope. Denoting by $r(\rho)$ the Hilbert-Schmidt distance between a state $\rho$ and the MMS,
\begin{equation}
r(\rho) = \lVert \rho - \rho_0 \rVert_{\mathrm{HS}} =  \sqrt{\mathrm{Tr}\left[\left(\rho-\rho_0\right)^2\right]},
\end{equation}
we have that all quantum states with $r(\rho)\leq r_{\mathrm{in}}^{W_\mathrm{min}}$ are AWB. This distance is equivalent to the Euclidean distance in the simplex between the spectra $\boldsymbol{\lambda}$ and $\boldsymbol{\lambda}_{0}$ of $\rho$ and the MMS respectively, i.e.,
\begin{equation*}
r(\rho) = \sqrt{\left(\sum_{i=0}^{2j}\lambda_{i}^{2}\right)-\frac{1}{2j+1}} = \lVert\boldsymbol{\lambda}-\boldsymbol{\lambda}_{0}\rVert.
\end{equation*}
In order to find the radius $r_{\mathrm{in}}^{W_{\mathrm{min}}}$ (see Fig.~\ref{fig:minimalpolytope} for $W_{\mathrm{min}}=0$) of the largest inner ball of the AWB polytope, we need to find the spectra on the hyperplanes of the AWB polytope with the minimum distance to the MMS. Mathematically, this translates in the following constrained minimization problem
\begin{equation}\label{eq:Minimization_rin}
\min_{\boldsymbol{\lambda}} \; \lVert\boldsymbol{\lambda}-\boldsymbol{\lambda}_0\rVert^2 \;\;\;
\text{ subject to } \left\{\begin{array}{l}
\sum_{i=0}^{2j}\lambda_{i}=1\\[8pt]
\boldsymbol{\lambda\cdot\Delta}=W_{\mathrm{min}}
\end{array}\right.
\end{equation}
where $\boldsymbol{\Delta}=\left(\Delta_{0},\Delta_{1},...,,\Delta_{2j}\right)$.
For this purpose, we use the method of Lagrange multipliers with the Lagrangian
\begin{equation*}
L = \lVert\boldsymbol{\lambda}-\boldsymbol{\lambda}_0\rVert^2+\mu_{1}\left(\boldsymbol{\lambda\cdot\Delta}-W_{\mathrm{min}}\right)+\mu_{2}\left(1-\sum_{i=0}^{2j}\lambda_{i}\right)
\end{equation*}
where $\mu_{1}, \mu_{2}$ are two Lagrange multipliers to be determined. The stationary points $\boldsymbol{\lambda}^*$ of the Lagrangian must satisfy the following condition 
\begin{equation}\label{eq:LagrangianDerivative}
\frac{\partial L}{\partial\boldsymbol{\lambda}}\Big|_{\boldsymbol{\lambda}=\boldsymbol{\lambda}^*} = \boldsymbol{0} \quad \Leftrightarrow \quad 2\boldsymbol{\lambda}^*+\mu_{1}\boldsymbol{\Delta}-\mu_{2}\boldsymbol{1} = \boldsymbol{0}
\end{equation}
with $\boldsymbol{1}=(1,1,...,1)$ of length $2j+1$. By summing over the components of \eqref{eq:LagrangianDerivative} and using Eq.~\eqref{eq:kernel_eigs_unit_sum}, we readily get
\begin{equation}\label{mu2mu1}
\mu_{2} = \frac{\mu_{1}+2}{2j+1}.
\end{equation}
Then, by taking the scalar product of \eqref{eq:LagrangianDerivative} with $\boldsymbol{\Delta}$ and using Eqs.~\eqref{identity2} and \eqref{mu2mu1}, we obtain
\begin{equation*}
\mu_{1} = \frac{1-(2j+1)W_{\mathrm{min}}}{2j(j+1)}\quad \mathrm{and} \quad \mu_{2} = \frac{(2j+1)-W_{\mathrm{min}}}{2j(j+1)}.
\end{equation*}
Finally, by substituting the above values for $\mu_1$ and $\mu_2$ in Eq.~\eqref{eq:LagrangianDerivative} and solving for the stationary point $\boldsymbol{\lambda}^*$, we get
\begin{equation}\label{rhostar}
\boldsymbol{\lambda}^*=\frac{\left[(2j+1)-W_{\mathrm{min}}\right]\boldsymbol{1}-\left[1-(2j+1)W_{\mathrm{min}}\right]\boldsymbol{\Delta}}{4j(j+1)}
\end{equation}
from which the inner ball radius follows as
\begin{equation*}
r_{\mathrm{in}}^{W_{\mathrm{min}}} = r(\rho^*) = \frac{1-(2j+1)W_{\mathrm{min}}}{2\sqrt{j(2j+1)(j+1)}}
\end{equation*} 
with $\rho^*$ a state with eigenspectrum \eqref{rhostar}.

\subsection{Smallest ball containing all AWB states}
\label{smallestouterball}
We give here the reasoning and evidence for a conjecture on the radius $r_{\mathrm{out}}^{W_\mathrm{min}}$ of the smallest outer ball of the polytope containing all AWB states. With the set of AWB states forming a convex polytope, $r_{\text{out}}^{W_\mathrm{min}}$ must be the radius associated with the outermost vertex.  Hence the problem is equivalent to finding this particular vertex within the minimal polytope.  For convenience let us call the matrix associated to any of the vertices a \textit{vertex state}.  Note, however, that for certain small values of $W_\mathrm{min}$ the polytopes may include trace-1 Hermitian matrices that are not positive semi-definite --- see Eq.\ \eqref{Wmincritical} and Fig.\ \ref{fig:criticalpolytope}.  This generalization however does not affect the analysis for this subsection.

In principle, the outermost vertex can always be determined on a case-by-case basis via the following procedure.  Recall from Sec.\ \ref{subsec:AWPPolytope} that an AWB vertex state with ordered spectrum $\boldsymbol{\lambda}^{\downarrow}$ is specified by $2j+1$ linear constraints on the eigenvalues. The first is normalization, the second is the AWB vertex criterion (i.e.,
Eq.\ \eqref{eq:ordered_awp_eq} with some $W_{\mathrm{min}}$), and the remaining $2j-1$ constraints come from a ($2j-1$)-sized sample from the $(2j)$-sized set of nearest-neighbour constraints \eqref{eq:set_of_NN_constraints}.  Thus the $2j$ states sitting on the $2j$ distinct vertices match up with the $\binom{2j}{2j-1} = 2j$ choices of bi-partitioning the ordered eigenvalues into a ``left'' set, $\boldsymbol{\omega}_n$, of size $n$ and a ``right'' set, $\boldsymbol{\sigma}_n$, of size $2j+1-n$, each of which contain eigenvalues of equal value $\omega_n$ and $\sigma_n$ respectively such that $\omega_n > \sigma_n$.  The full eigenspectrum is the concatenation $\boldsymbol{\lambda}^{\downarrow}_{\mathrm{v}_n} = \boldsymbol{\omega}_n \circ \boldsymbol{\sigma}_n$, and normalization becomes
\begin{equation}\label{eq:vertex_state_normalization}
    n\omega_n + (2j+1-n)\sigma_n = 1, \quad n \in \{ 1,...,2j \}.
\end{equation}
As we are temporarily allowing the ordered spectrum $\boldsymbol{\lambda}^{\downarrow}$ to have negative components, Eq.\ \eqref{eq:vertex_state_normalization} should be interpreted only as requiring the vertices to lie in the affine span generated by the state simplex (i.e., not necessarily within the simplex). Inserting $\boldsymbol{\lambda}^{\downarrow}_{\mathrm{v}_n}$ and \eqref{eq:vertex_state_normalization} into the AWB vertex criterion the weights $\omega_n$ can be solved as a function of the kernel eigenvalues and $W_\mathrm{min}$:
\begin{align}\label{eq:omega_n_explicit}
    \omega_n &= \frac{\sum_{i=n}^{2j} \Delta_i^\uparrow - (2j+1-n) W_\mathrm{min}}{ n \sum_{i=n}^{2j} \Delta_i^\uparrow - (2j+1-n)\sum_{i=0}^{n-1} \Delta_i^\uparrow } \nonumber \\
    &= \frac{\tau_n - (2j+1 - n) W_\mathrm{min}}{(2j+1)\tau_n - (2j+1-n)}
\end{align}
where in the second line we used the unit-trace property \eqref{eq:kernel_eigs_unit_sum} of the kernel and 
\begin{equation}
    \tau_n = \sum_{i=n}^{2j} \Delta_i^\uparrow = \sum_{i=0}^{2j-n} \Delta_i^\downarrow
\end{equation}
is the sum over the largest $2j+1-n$ kernel eigenvalues.  The purity $\gamma_{\mathrm{v}_n}$ and distance $r_{\mathrm{v}_n}$ of the $n$-th vertex is then given by
\begin{align}
    \gamma_{\mathrm{v}_n} &= n \omega_n^2 + (2j+1-n)\sigma_n^2 \\
    r_{\mathrm{v}_n} &= \sqrt{\gamma_{\mathrm{v}_n} - \frac{1}{2j+1}},
\end{align}
which are functions of only the kernel eigenvalues and $W_\mathrm{min}$.  Note that purity, being defined as the sum of squares of the eigenvalues, remains a faithful notion of distance to the MMS even when such spectra are allowed to go negative.  After computing each of these numbers, $r_{\text{out}}^{W_\mathrm{min}}$ would correspond to the largest one, and the set of states satisfying this condition would be the intersection of the associated ball with the state simplex.  In Sec.\ \ref{sec:spin1} we present details of this procedure for $j=1$ and $W_\mathrm{min} = 0$.

Despite this somewhat involved procedure, we numerically find it is always the case that the first vertex, $\mathrm{v}_1$, remains within the state simplex for all $W_\mathrm{min} \in [\Delta^\uparrow_0,\frac{1}{2j+1}]$ and, relatedly, that
\begin{equation}
    r^{W_\mathrm{min}}_{\text{out}} = r_{\text{v}_1}.
\end{equation}
We conjecture this to be true in all finite dimensions.  Part of the difficulty in proving this in general comes from the non-trivial nature of the kernel eigenvalues \eqref{eq:kernel_eigenvalues} and from further numerical evidence suggesting that no vertex state ever majorizes any other vertex state.

Furthermore, with the most negative kernel eigenvalue \eqref{eq:kernel_eigenvalue_assumption} being $\Delta^\uparrow_0 = \Delta_{j,j-1}$, the vertex state $\rho_{\text{v}_1}$ takes the special form 
\begin{equation}\label{eq:outer_vertex_state}
    \omega_1 \ketbra{j,j-1}{j,j-1} + \frac{1-\omega_1}{2j} \sum_{m\neq j-1} \ketbra{j,m}{j,m}
\end{equation}
where
\begin{align}
    \omega_1 &= \frac{\sum_{m\neq j-1}\Delta_{j,m} - 2j W_{\mathrm{min}} }{\sum_{m\neq j-1}\Delta_{j,m} - 2j\Delta_{j,j-1}} \nonumber\\
    &= \frac{1-\Delta_{j,j-1} - 2jW_{\mathrm{min}}}{1-(2j+1)\Delta_{j,j-1}}. \label{eq:outer_vertex_weight}
\end{align}
The minimal outer radius $r_{\text{out}}^{W_\mathrm{min}}$ is then conjectured to be
\begin{align}
    r_{\text{out}}^{W_\mathrm{min}} &= \sqrt{ \gamma_{\text{v}_1} - \frac{1}{2j+1} } \nonumber\\
    &= \sqrt{\omega_1^2 + 2j\left( \frac{1-\omega_1}{2j} \right)^2 - \frac{1}{2j+1}} \nonumber \\
    &= \sqrt{\frac{2j}{2j+1}} \left\lvert \frac{W_{\mathrm{min}}(2j+1) - 1}{\Delta_{j,j-1}(2j+1) - 1} \right\rvert \label{apndx_outer_radius} . 
\end{align}

The radius \eqref{apndx_outer_radius} can be seen as a scaled factor of $\sqrt{2j/(2j+1)}$, the distance from any pure state to the maximally mixed state.  When the highest cutoff is set, $W_{\mathrm{min}} = 1/(2j+1)$, this outer radius vanishes as the only state that can satisfy the cutoff is the maximally mixed state, which has zero distance to itself.  When the lowest cutoff is set, $W_{\mathrm{min}}=\Delta_{j,j-1}$, the scaling factor in \eqref{apndx_outer_radius} becomes unity and the outer radius reduces to the distance to pure states, which reflects the fact that now the entire simplex (and hence all mixed states) is contained within the AWB polytope.

An operational interpretation of this radius is available by noting that the multiqubit realization of the $\ket{j,j-1}$ state, which has the most pointwise-negative Wigner function allowable (occurring at the North pole), is in fact the $W$ state introduced in the context of LOCC entanglement classification \cite{Dur_LOCC_2000}.  And since the maximally mixed state has uniform eigenvalues, Eq.\ \eqref{eq:outer_vertex_state} may be interpreted as the end result of mixing the $W$ state with the maximally mixed state until the Wigner function at the North pole hits $W_\mathrm{min}$.  The distance between the resulting state and the maximally mixed state is exactly our conjectured $r_{\text{out}}^{W_\mathrm{min}}$.  In particular, when the Wigner function vanishes at the North pole, the radius reduces to a tight, purity-based, necessary condition to be AWP.

\bibliographystyle{quantum}
\bibliography{references}

\end{document}